\makeatletter
\def\input@path{{styles/}{../styles/}}
\makeatother

\ifx\SODA\undefined%
   \providecommand{\SODAVer}[1]{}%
   \providecommand{\NotSODAVer}[1]{#1}%
   \newcommand{\RegVer}[1]{#1}
\else%
   \providecommand{\SODAVer}[1]{#1}%
   \providecommand{\NotSODAVer}[1]{}%
   \newcommand{\RegVer}[1]{}
\fi

\def\UseBibLatex{1}

\SODAVer{%
   \documentclass[twoside,leqno]{article}%
   \usepackage{ltexpprt}%
   \newcommand{\tcite}[1]{\cite{#1}}%
   \usepackage[letterpaper]{geometry}

}

\NotSODAVer{%
    \documentclass[11pt]{article}%
    \usepackage[in]{fullpage}%
}

\usepackage{euscript}%
\usepackage{amsfonts}%
\usepackage{xcolor}%
\usepackage{graphicx}%
\usepackage{xspace}%
\SODAVer{%
   \usepackage{float}%
}%
\usepackage{hyperref}%

\hypersetup{%
   breaklinks,%
   colorlinks=true,%
   urlcolor=[rgb]{0.25,0.0,0.0},%
   linkcolor=[rgb]{0.5,0.0,0.0},%
      citecolor=[rgb]{0,0.2,0.445},%
      filecolor=[rgb]{0,0,0.4},
      anchorcolor=[rgb]={0.0,0.1,0.2}%
   }

\usepackage{algorithm}

\newcommand{\eps}{\varepsilon}%

\def\C{C}%
\def\F{\mathcal{F}}%
\def\J{\mathcal{J}}%
\def\K{\mathcal{K}}
\def\R{\mathcal{R}\flogi}
\def\H{\mathcal{H}}
\def\W{\mathcal{W}}
\def\O{\widetilde{O}}

\newcommand{\weight}{\AlgorithmI{weight}\xspace}%
\newcommand{\double}{\AlgorithmI{double}\xspace}%

\newcommand{\Insert}{\AlgorithmI{insert}\xspace}%
\newcommand{\Delete}{\AlgorithmI{delete}\xspace}%

\newcommand{\RSample}{\Mh{\mathcal{R}}}%

\newcommand{\sample}{\AlgorithmI{sample}\xspace}%

\newcommand{\halve}{\AlgorithmI{halve}\xspace}%

\usepackage{amsmath}%
\usepackage{amssymb}%

\NotSODAVer{%
   \usepackage[amsmath,thmmarks]{ntheorem}%
   \theoremseparator{.}%

   \usepackage{titlesec}%
   \titlelabel{\thetitle. }%

   \theoremseparator{.}%

   \theoremstyle{plain}%
   \newtheorem{theorem}{Theorem}[section]

   \newtheorem{lemma}[theorem]{Lemma}
   
   \newtheorem{corollary}[theorem]{Corollary}

   \newtheorem{observation}[theorem]{Observation}

   \theoremstyle{plain}%
   \theoremheaderfont{\sf} \theorembodyfont{\upshape}%
   \newtheorem*{remark:unnumbered}[theorem]{Remark}%
   \newtheorem{remark}[theorem]{Remark}%
   
   \newtheorem{defn}[theorem]{Definition}

   \newcommand{\myqedsymbol}{\rule{2mm}{2mm}}

   \theoremheaderfont{\em}%
   \theorembodyfont{\upshape}%
   \theoremstyle{nonumberplain}%
   \theoremseparator{}%
   \theoremsymbol{\myqedsymbol}%
   \newtheorem{proof}{Proof:}%

}
\RegVer{%
   \numberwithin{figure}{section}%
   \numberwithin{table}{section}%
   \numberwithin{equation}{section}%
}%

\SODAVer{%
   \newtheorem{defn}{Definition}[section]
   \newtheorem{remark}{Remark}[section]
}
\usepackage{mleftright}

\definecolor{blue25emph}{rgb}{0, 0, 11}
\providecommand{\emphic}[2]{%
   \textcolor{blue25emph}{%
      \textbf{\emph{#1}}}%
   \index{#2}}

\providecommand{\emphi}[1]{\emphic{#1}{#1}}

\definecolor{almostblack}{rgb}{0, 0, 0.3}
\providecommand{\emphw}[1]{{\textcolor{almostblack}{\emph{#1}}}}%

\newcommand{\atgen}{\symbol{'100}} \newcommand{\SarielThanks}[1]{%
   \thanks{%
      Department of Computer Science; University of Illinois; 201
      N. Goodwin Avenue; Urbana, IL, 61801, USA; {\tt
         sariel\atgen{}illinois.edu};
      {\tt \url{http://sarielhp.org/}.} #1
   }%
}

\newcommand{\PankajThanks}[1]{%
   \thanks{%
      Department of Computer Science, Duke University, Durham NC
      27708. %
      #1
   }%
}
\newcommand{\RahulThanks}[1]{%
   \thanks{%
      Department of Computer Science, Duke University, Durham NC
      27708. %
      #1
   }%
}
\newcommand{\StavrosThanks}[1]{%
   \thanks{%
        Department of Computer Science. University of Illinois at Chicago, Chicago IL 60607. %
      #1
   }%
}

\newcommand{\HLink}[2]{\hyperref[#2]{#1~\ref*{#2}}}
\newcommand{\HLinkSuffix}[3]{\hyperref[#2]{#1\ref*{#2}{#3}}}

\newcommand{\figlab}[1]{\label{fig:#1}}
\newcommand{\figref}[1]{\HLink{Figure}{fig:#1}}

\newcommand{\thmlab}[1]{{\label{theo:#1}}}
\newcommand{\thmref}[1]{\HLink{Theorem}{theo:#1}}

\newcommand{\corlab}[1]{\label{cor:#1}}
\newcommand{\corref}[1]{\HLink{Corollary}{cor:#1}}%

\newcommand{\seclab}[1]{\label{sec:#1}}
\newcommand{\secref}[1]{\HLink{Section}{sec:#1}}

\newcommand{\itemlab}[1]{\label{item:#1}}

\newcommand{\itemrefY}[2]{\hyperref[item:#2]{#1}}

\providecommand{\deflab}[1]{\label{def:#1}}
\newcommand{\defref}[1]{\HLink{Definition}{def:#1}}

\newcommand{\remlab}[1]{\label{rem:#1}}
\newcommand{\remref}[1]{\HLink{Remark}{rem:#1}}%

\newcommand{\lemlab}[1]{\label{lemma:#1}}
\newcommand{\lemref}[1]{\HLink{Lemma}{lemma:#1}}%

\providecommand{\eqlab}[1]{}%
\renewcommand{\eqlab}[1]{\label{equation:#1}}
\newcommand{\Eqref}[1]{\HLinkSuffix{Eq.~(}{equation:#1}{)}}

\newcommand{\remove}[1]{}%
\newcommand{\Set}[2]{\left\{ #1 \;\middle\vert\; #2 \right\}}
\newcommand{\pth}[2][\!]{\mleft({#2}\mright)}%
\newcommand{\pbrcx}[1]{\left[ {#1} \right]}%

\newcommand{\ProbLTR}{\mathbb{P}}%
\newcommand{\Prob}[1]{\mathop{\ProbLTR} \mleft[ #1 \mright]}%

\newcommand{\Ex}[2][\!]{\mathop{\mathbb{E}}#1\pbrcx{#2}}

\newcommand{\ceil}[1]{\left\lceil {#1} \right\rceil}
\newcommand{\floor}[1]{\left\lfloor {#1} \right\rfloor}

\newcommand{\cardin}[1]{\left| {#1} \right|}%

\renewcommand{\th}{th\xspace}

\newcommand{\ZZ}{\mathbb{Z}}%
\renewcommand{\Re}{\mathbb{R}}%

\usepackage[inline]{enumitem}

\newlist{compactenumA}{enumerate}{5}%
\setlist[compactenumA]{topsep=0pt,itemsep=-1ex,partopsep=1ex,parsep=1ex,%
   label=(\Alph*)}%

\newlist{compactenuma}{enumerate}{5}%
\setlist[compactenuma]{topsep=0pt,itemsep=-1ex,partopsep=1ex,parsep=1ex,%
   label=(\alph*)}%

\newlist{compactenumI}{enumerate}{5}%
\setlist[compactenumI]{topsep=0pt,itemsep=-1ex,partopsep=1ex,parsep=1ex,%
   label=(\Roman*)}%

\newlist{compactenumi}{enumerate}{5}%
\setlist[compactenumi]{topsep=0pt,itemsep=-1ex,partopsep=1ex,parsep=1ex,%
   label=(\roman*)}%

\newlist{compactitem}{itemize}{5}%
\setlist[compactitem]{topsep=0pt,itemsep=-1ex,partopsep=1ex,parsep=1ex,%
   label=\ensuremath{\bullet}}%

\providecommand{\BibLatexMode}[1]{}
\providecommand{\BibTexMode}[1]{#1}

\ifx\UseBibLatex\undefined%
  \renewcommand{\BibLatexMode}[1]{}
  \renewcommand{\BibTexMode}[1]{#1}
\else
  \renewcommand{\BibLatexMode}[1]{#1}
  \renewcommand{\BibTexMode}[1]{}
\fi

\newcommand{\UsePackage}[1]{%
  \IfFileExists{styles/#1.sty}{%
      \usepackage{styles/#1}%
   }{%
      \IfFileExists{../styles/#1.sty}{%
         \usepackage{../styles/#1}%
      }{%
         \usepackage{#1}%
      }%
   }%
}

\BibLatexMode{%
   \usepackage[bibencoding=utf8,style=alphabetic,backend=biber]{biblatex}%
   \UsePackage{sariel_biblatex}%
}

\providecommand{\Mh}[1]{#1}%

\newcommand{\ISX}[1]{\mathcalb{i}\pth{#1}} %
\newcommand{\fIS}{\mathcalb{i}^*} %

\newcommand{\fPN}{\mathcalb{p}^{*}} %
\newcommand{\PNX}[1]{\Mh{\mathcalb{p}}\pth{#1}}

\newcommand{\popt}{\Mh{\mathcalb{p}}}

\newcommand{\AlgPierce}{\AlgorithmI{Pierce}\xspace}%
\newcommand{\TPierce}{\ensuremath{T_\AlgorithmI{P}\xspace}}%
\newcommand{\Partition}{\Mh{\Pi}}
\newcommand{\Term}[1]{\textsf{#1}}

\newcommand{\LCA}{\Term{LCA}\xspace}
\newcommand{\LP}{\Term{LP}\xspace}
\newcommand{\MWU}{\Term{MWU}\xspace}
\newcommand{\ts}{\hspace{0.6pt}}

\newcommand{\rstY}[2]{#1 \sqcap  #2}%
\newcommand{\nrstY}[2]{\cardin{\rstY{#1}{#2}}}%

\renewcommand{\S}{\Mh{\mathcal{S}}}%
\renewcommand{\L}{\Mh{\mathcal{L}}}%
\newcommand{\U}{\Mh{\mathcal{U}}}%
\newcommand{\I}{\Mh{\mathcal{I}}}

\newcommand{\projY}[2]{#1_{\downarrow #2}}%
\newcommand{\XX}{\Mh{\mathcal{X}_{\varphi}}}
\newcommand{\YY}{\Mh{\mathcal{Y}_{\varphi}}}

\newcommand{\XXz}{\Mh{\mathcal{X}_{z}}}
\newcommand{\YYz}{\Mh{\mathcal{Y}_{z}}}

\newcommand{\MX}{\Mh{\mathcal{X}}}
\newcommand{\MY}{\Mh{\mathcal{Y}}}

\newcommand{\interior}{\mathrm{int}}
\newcommand{\intX}[1]{\interior\pth{#1}}%

\newcommand{\e}{e}

\newcommand{\si}[1]{#1}
\newcommand{\N}{\Mh{\mathcal{N}}}

\newcommand{\X}{\Mh{{X}}}

\newcommand{\Q}{\Mh{{Q}}}

\renewcommand{\P}{\Mh{{P}}}

\RegVer{%
   \titleformat{\paragraph}[runin]
   {\normalfont\bfseries}
   {\theparagraph}
   {1em}
   {\addperiod}

   \newcommand{\addperiod}[1]{#1.}

   \titlespacing{\paragraph}{0pt}{2ex}{0.2cm}%
}

\DeclareFontFamily{U}{BOONDOX-calo}{\skewchar\font=45 }
\DeclareFontShape{U}{BOONDOX-calo}{m}{n}{
  <-> s*[1.05] BOONDOX-r-calo}{}
\DeclareFontShape{U}{BOONDOX-calo}{b}{n}{
  <-> s*[1.05] BOONDOX-b-calo}{}
\DeclareMathAlphabet{\mathcalb}{U}{BOONDOX-calo}{m}{n}
\SetMathAlphabet{\mathcalb}{bold}{U}{BOONDOX-calo}{b}{n}
\DeclareMathAlphabet{\mathbcalb}{U}{BOONDOX-calo}{b}{n}

\providecommand{\TPDF}[2]{\texorpdfstring{#1}{#2}}

\newcommand{\Dim}{\Mh{\mathcalb{d}}}%
\newcommand{\BadProb}{\Mh{\varphi}}%

\newcommand{\GroundSet}{\mathcal{X}}

\newcommand{\FGroundSet}{\mathsf{X}}
\newcommand{\Family}{\Mh{\mathcal{F}}}%

\newcommand{\range}{\Mh{\mathbf{r}}}
\newcommand{\bb}{\Mh{\mathsf{b}}}%
\newcommand{\bbA}{\Mh{\mathsf{c}}}%
\newcommand{\bbC}{\Mh{\mathsf{g}}}%
\newcommand{\bbD}{\Mh{\mathsf{e}}}%

\newcommand{\MeasureChar}{\overline{m}}

\newcommand{\Measure}[1]{\MeasureChar\pth{#1}}

\newcommand{\sMeasure}[1]{\overline{s}\pth{#1}}

\newcommand{\Grid}{G}%
\newcommand{\GridSet}{\mathcal{G}}%

\newcommand{\Arr}{\Mh{\mathop{\mathrm{\mathcal{A}}}}}
\newcommand{\ArrX}[1]{\Arr\pth{#1}}%

\newcommand{\etal}{\textit{et~al.}\xspace}

\newcommand{\Edges}{\mathsf{E}}%
\providecommand{\Bronnimann}{Br{\"o}nnimann\xspace}

\DefineNamedColor{named}{AlgorithmColor}{cmyk}{0.07,0.90,0,0.34}

\newcommand{\AlgorithmI}[1]{{%
      \textcolor[named]{AlgorithmColor}{\texttt{\bf{#1}}}%
   }}

\newcommand{\ii}{\mathcalb{i}}

\usepackage{stmaryrd}%
\providecommand{\IntRange}[1]{\mleft\llbracket #1 \mright\rrbracket}
\newcommand{\IRX}[1]{\IntRange{#1}}%
\newcommand{\IRY}[2]{\left\llbracket #1:#2 \right\rrbracket}
\newcommand{\nQ}{\Mh{\mathcalb{n}}}%
\newcommand{\nG}{\Mh{n_\GridSet}}%
\renewcommand{\R}{\Mh{\mathcal{R}}}%

\newcommand{\RangeSet}{{\Mh{\mathcal{R}}}}

\newcommand{\rect}{\Mh{\mathsf{b}}}

\newcommand{\cell}{\varphi}%

\newcommand{\Sample}{\Mh{\mathcal{S}}}%

\newcommand{\CBOXES}{\mathcal{C}}%
\newcommand{\CBeps}{\mathcal{C}_{\eps/2}}%
\newcommand{\BBGS}{\Mh{\mathcal{B}}_\GridSet}%
\newcommand{\BBA}{\Mh{\widehat{\mathcal{C}}}}%
\newcommand{\BBB}{\Mh{{\mathcal{D}}}}%
\newcommand{\BBM}{\Mh{{\mathcal{M}}}}%
\newcommand{\Boxes}{\Mh{\mathcal{B}}}%
\newcommand{\BoxesB}{\Mh{\mathcal{C}}}%
\newcommand{\AB}{\Mh{{\mathcal{C}}}}%

\newcommand{\hBoxes}{{\Boxes_\downarrow}}

\newcommand{\ListB}{\Mh{\mathcal{S}}}%
\newcommand{\LightB}{\Mh{\mathcal{L}}}%
\newcommand{\HeavyB}{\Mh{\mathcal{H}}}%
\newcommand{\IndepB}{\Mh{\mathcal{I}}}%
\newcommand{\Crates}{\Mh{\mathcal{T}}}%
\newcommand{\lgEps}{\Mh{L}}%
\newcommand{\nMass}{\Mh{\mathcalb{w}}}%

\newcommand{\DS}{\Mh{\EuScript{D}}}%

\newcommand{\nRounds}{\Mh{\zeta}}%

\newcommand{\polylog}{\mathop{\mathrm{polylog}}}%
\newcommand{\VV}{\Mh{\mathcal{V}}}%
\newcommand{\VX}[1]{\VV\pth{#1}}

\newcommand{\wC}{\Mh{\omega}}
\newcommand{\wX}[1]{\wC\pth{#1}}%

\newcommand{\hwC}{\Mh{\mathcalb{w}}}

\newcommand{\hwY}[2]{\hwC^{}_{\!#2}\pth{#1}}

\newcommand{\mm}{\mathcalb{m}\!\ts}%

\newcommand{\Tree}{\Mh{T}}%

\newcommand{\BS}{\mathcal{B}}
\newcommand{\BSB}{\mathcal{C}}
\newcommand{\BSC}{\mathcal{D}}

\newcommand{\CFaceX}[1]{\mathcal{F}(#1)}%

\renewcommand{\Re}{\mathbb{R}}%

\newcommand{\Ps}{\mathcal{P}}

\newcommand{\VSetX}[1]{\mathcal{V}^*\pth{#1}}

\NotSODAVer{%
   \ifx\sarielComp\undefined%
   \else
       \usepackage{sariel_colors}%
   \fi
}

\newcommand{\dY}[2]{\left\| #1 - #2 \right\|}

\newcommand{\codim}{\mathrm{codim}}%
\newcommand{\poly}{\mathrm{poly}}%

\usepackage[normalem]{ulem}

\newcommand{\cenX}[1]{\overline{\mathsf{c}}\pth{#1}}
\newcommand{\anc}{\mathrm{Anc}}
\UsePackage{arxiv_logo}

\providecommand{\tcite}[1]{\cite{#1}}

\BibLatexMode{\bibliography{rect_net}}%

\title{Fast Approximation Algorithms for Piercing Boxes by Points%
   \thanks{A full version of this paper is available on the \arXiv
      \cite{ahrs-faapb-23}. A preliminary version of this paper
      appeared in SODA 2024 \cite{ahrs-faapb-24}.}%
}

\RegVer{\date{\today}}

\author{%
   Pankaj K. Agarwal%
   \PankajThanks{Work by Pankaj Agarwal was partially supported by NSF
      grants \si{IIS-18-14493}, CCF-20-07556, and CCF-22-23870.}%
   \and%
   Sariel Har-Peled%
   \SarielThanks{Work by Sariel Har-Peled was partially supported by a NSF
      AF award CCF-1907400 and CCF-2317241.  }%
   \and%
   Rahul Raychaudhury%
   \RahulThanks{}%
   \and%
   Stavros Sintos%
   \StavrosThanks{}%
}%

\begin{document}
\date{\today}%
\maketitle

\begin{abstract}
    Let $\Boxes=\{\bb_1, \ldots ,\bb_n\}$ be a set of $n$ axis-aligned
    boxes in $\Re^d$ where $d\geq2$ is a constant. The \emph{piercing
       problem} is to compute a smallest set of points
    $\N \subset \Re^d$ that hits every box in $\Boxes$, i.e.,
    $\N\cap \bb_i\neq \emptyset$, for $i=1,\ldots, n$. Let
    $\popt=\popt(\Boxes)$, the \emph{piercing number} be the minimum
    size of a piercing set of $\Boxes$. We present a randomized
    $O(d^2\log\log \popt)$-approximation algorithm with expected
    running time $O(n^{d/2}\polylog n)$. Next, we present a faster
    $O(n^{\log d+1})$-time algorithm but with a slightly inferior
    approximation factor of $O(2^{4d}\log\log\popt)$.
    The running time of both  algorithms can
    be improved to near-linear using a sampling-based technique, if
    $\popt = O(n^{1/d})$.

    For the dynamic version of the problem in the plane, we obtain a
    randomized $O(\log\log\popt)$-approximation algorithm with
    $O(n^{1/2}\polylog n )$ amortized expected update time for
    insertion or deletion of boxes. For squares in $\Re^2$, the update
    time can be improved to $O(n^{1/3}\polylog n )$.

\end{abstract}

\SODAVer{%
  \fancyfoot[R]{\scriptsize{Copyright \textcopyright\ 2024 by SIAM\\
  Unauthorized reproduction of this article is prohibited}}
}

\section{Introduction}

\paragraph{Problem statement}

A \emphi{box} is an axis-aligned box in $\Re^d$ of the form
$\prod_{i=1}^d [\alpha_i, \beta_i]$. A one dimensional box is an
\emphi{interval}, and a two dimensional box is a \emphi{rectangle}.
Let $\Boxes=\{\bb_1, \ldots ,\bb_n\}$ be a set of $n$ boxes in
$\Re^d$. A subset $\N\subset \Re^d$ is a \emphw{piercing set} of
$\Boxes$ if $\N\cap R\neq \emptyset$ for every box $\bb\in
\Boxes$. The \emph{piercing problem} asks to find a piercing set $\N$
of $\Boxes$ of the smallest size, and the \emphi{piercing number} of
$\Boxes$ is denoted by $\popt=\PNX{\Boxes}$.
The piercing problem is a
fundamental problem in computational geometry and has applications in
facility location, sensor networks, etc.

The piercing problem is closely related to the classical
\emphw{geometric hitting-set} problem: Given a (geometric) range space
$\Sigma=(\X, \RangeSet)$, where $\X$ is a (finite or infinite) set of
points in $\Re^d$ and $\RangeSet \subseteq 2^\X$ is a finite family of
ranges, defined by simply-shaped regions such as rectangles, balls,
hyperplanes etc. That is, each range in $\Sigma$ is of the form
$R\cap \X$, where $R\in\RangeSet$ is a geometric
shape\footnote{Strictly speaking,
   $\Sigma=(\X,\Set{R\cap\X}{ R\in \RangeSet})$ but with a slight
   abuse of notation we use $(\X,\RangeSet)$ to denote the range
   space.}. A subset $H\subseteq \X$ is a \emph{hitting set} of
$\Sigma$ if $H\cap R\neq \emptyset$ for all $R \in \RangeSet$. The
hitting-set problem asks for computing a minimum-size hitting set of
$\Sigma$ and the smallest size is referred to as the \emph{stabbing
   number} of $\Sigma$, denoted by $\popt(\X,\RangeSet)$. The piercing
problem is a special case of the hitting-set problem in which
$\Sigma=(\Re^d,\Boxes)$. Instead of letting the set of points be the
entire $\Re^d$, we can choose the set of points to be the set of
vertices in $\ArrX{\Boxes}$, the arrangement of $\Boxes$\footnote{The
   arrangement of $\Boxes$, denoted by $\ArrX{\Boxes}$, is the
   partition of $\Re^d$ into maximal connected cells so that all
   points within each cell are in the interior/boundary of the same
   set of rectangles. It is well known that $\ArrX{\Boxes}$ has
   $O(n^d)$ complexity~\cite{e-acg-87}.}, and the
range space is now
$\Sigma=\bigl(\VV,\Set{ \bb \cap \VV }{ \bb\in\Boxes } \bigr)$, where
$\VV=\VX{\Boxes}$ is the set of vertices in $\ArrX{\Boxes}$. It is
easily seen that $\Boxes$ has a piercing set of size $\popt$ if and
only if $\Sigma$ has a hitting set of size $\popt$. Hitting set for
general range spaces was listed as one of the original NP-complete
problems \cite{k-racp-72}. Furthermore, the box piercing problem is
NP-Complete even in 2D \cite{fpt-opcpa-81}, so our goal is to develop
an efficient approximation algorithm for the piercing problem.

In many applications, especially those dealing with large data sets,
simply a polynomial-time algorithm is not enough, and one desires an
algorithm whose running time is near-linear in $|\Boxes|$. In
principle, the classical greedy algorithm can be applied to the range
space $(\VV,\Boxes)$, but $|\VV|=O(n^d)$, so it does not lead to a
fast algorithm. Intuitively, due to the unconstrained choice of points
with which to pierce boxes of $\Boxes$, the piercing problem seems
easier than the geometric hitting-set problem and should admit faster
and better approximations. In this paper, we make progress towards
this goal for a set of boxes in both static and dynamic settings.

\paragraph{Related work}

The well-known shifting technique by Hochbaum and Maass
\cite{hm-ascpp-85} can be used to obtain a PTAS when $\Boxes$
comprises of unit-squares or near-equal-sized fat objects in any fixed
dimension. Efrat \etal \cite{ekns-ddsfo-00} designed an
$O(1)$-approximation algorithm for a set of arbitrary ``fat'' objects
that runs in near-linear time in 2d and 3d. Chan \cite{c-paspp-03}
gave a separator-based PTAS for arbitrary sized fat objects, with
running time $O(n^{O(1/\eps^{d})})$. Chan and
Mahmood~\cite{cm-apnur-05} later gave a PTAS for a set of boxes with
arbitrary width but unit height. All of the above results consider a
restricted setting of boxes. Surprisingly, little is known about the
piercing problem for a set of arbitrary axis-aligned boxes in
$\Re^d$. By running a greedy algorithm or its variants based on a
multiplicative weight update (\MWU) method, an
$O(\log\popt)$-approximation algorithm with running time roughly
$O( n^d )$ can be obtained for the box-piercing problem in
$\Re^d$. Using the weak $\eps$-net result by Ezra \cite{e-nawea-10}
(see also \cite{aes-ssena-10}), the approximation factor improves to
$O(\log\log \popt)$. An interesting question is what is the smallest
piercing set one can find in near-linear time. Nielsen
\cite{n-fsbhd-00} presented an $O(\log^{d-1} \popt)$-approximation
divide-and-conquer algorithm that runs in $O(n\log^{d-1} n)$ time. We
are unaware of any near-linear time algorithm even with
$O(\log \popt)$-approximation ratio for $d\geq 3$.  We note that the
piercing problem has also been studied in discrete and convex
geometry, where the goal is to bound the size of the piercing set for
a family of objects with certain properties. See e.g. \cite{ak-pcs-92,
   csz-pab-18}.

There has been much work on the geometric hitting-set problem. For a
range space $\Sigma=(\X,\RangeSet)$ and weight function
$\wC: \X \xrightarrow{}\Re_{\geq 0}$, a subset $\N\subseteq \X$ is an
$\eps$-net if for any $R \in \RangeSet$ with
$\wX{R}\geq \eps \wX{\X}$, we have $R \cap \N \neq \emptyset$. The
multiplicative weight update (\MWU) method assigns a weight to each
point so that every range in $\RangeSet$ becomes
``${1}/{e \popt}$-heavy'' and then one simply chooses a
${1}/{e \popt}$-net.  Using the \MWU method and results on
$\eps$-nets, \Bronnimann and Goodrich \cite{bg-aoscf-95} presented a
polynomial-time $O(\log\popt)$-approximation algorithm for the
hitting-set problem for range spaces with finite VC-dimension
\cite{vc-ucrfe-71}. Later Agarwal and Pan \cite{ap-naghs-20} presented
a more efficient implementation of the \MWU method for geometric range
spaces. Their approach led to
$O\bigl((|X|+|\RangeSet|)\polylog n \bigr)$ algorithm with
$O(\log \popt)$-approximation for many cases including a set of
rectangles in $\Re^d$ (see also \cite{bmr-peagh-18, ch-faags-20}). But
it does not lead to a near-linear time algorithm for the piercing
problem because $|X|=O(n^d)$ in this case. In another line of work,
polynomial-time approximation algorithms for hitting sets based on
local search have also been proposed \cite{mr-irghs-10}.

The \MWU algorithm essentially solves and rounds the \LP associated
with the hitting-set problem, see \cite[Chapter 6]{h-gaa-11}. Thus,
the approximability of the problem is strongly connected to the
integrality gap of the \LP. For the hitting-set problem of points with
boxes for $d\leq 3$, Aronov \etal \cite{aes-ssena-10} showed a
rounding scheme with integrality gap $O( \log \log
\popt)$. Furthermore, Ezra \cite{e-nawea-10} showed the same gap holds
in higher dimensions if one is allowed to use any point to do the
piercing. Surprisingly, Pach and Tardos \cite{pt-tlbse-11} showed that
this integrality gap is tight. While a better approximation than the
integrality gap can be obtained in a few cases \cite{ch-aamis-12,
   mr-irghs-10}, these algorithms require a fundamentally different
approach. Thus, a major open problem is to obtain an
$O(1)$-approximation algorithm for the box-piercing problem.

Recently, Agarwal \etal \cite{acsxx-dgsch-22} initiated a study of
dynamic algorithms for geometric instances of set-cover and
hitting-set. Here the focus is on maintaining an approximately optimal
hitting-set (resp. set-cover) of a dynamically evolving instance,
where in each step a new object may be added or deleted. They
introduced fully dynamic sublinear time hitting-set algorithms for
squares and intervals.  These results were improved and generalized in
\cite{chsx-dgscr-22,ch-mddsg-21}.  Khan \etal \cite{klssw-odags-23}
proposed a dynamic data structure for maintaining a
$O(\polylog n)$-factor approximation of the optimal hitting-set for
boxes under restricted settings, but no algorithm with sublinear
update time for the general setting is known.

\paragraph{Dependency on dimension}

Bounds stated with $O( \cdot)$ hides constants that depend on the
dimension $d$, and this dependency is almost always exponential (or
even doubly exponential). Where the exact dependency on the dimension
matters, we explicitly state it in the bound.

\paragraph{Our results}

In this paper, we design $O(\log\log\popt)$ approximation algorithms
for the box-piercing problem.  A naive way to get a
$O( 2^{O(d)} \log\log \popt )$-approximation is by using
aforementioned \MWU \cite{ap-naghs-20, bg-aoscf-95} based hitting-set
algorithms on the range space $(\VV,\Boxes)$, and use Ezra's
\cite{e-nawea-10} algorithm for computing a weak $\eps$-net instead of
computing a strong $\eps$-net.  While this naive approach gives the
desired approximation, it runs in $\Omega(n^d)$ time. We present
significantly faster $O(\log\log\popt)$ approximation algorithms.

\subparagraph{(A) New constructions of (weak) $\eps$-nets.}

We revisit Ezra's result \cite{e-nawea-10} on weak $\eps$-net, and
improve its dependency on the dimension, and in the process simplify
it, see \secref{weak-net}.  In particular, given a set $\P$ of $n$
points in $\Re^d$, a weight function $w:\P\xrightarrow{}\Re_{\geq 0}$,
and a parameter $\eps\in (0,1)$, the new algorithm computes a weak
$\eps$-net $\N\subseteq \Re^d$ of $O(d^2 \eps^{-1}\log\log \eps^{-1})$
points (not necessarily a subset of $\P$) in
$O( n+\eps^{-1} \log^{O(d^2)} \eps^{-1} )$ expected time, such that
$\bb\cap \N\neq\emptyset$ for all boxes $\bb$ with
$\wX{\bb \cap \P} \geq \eps \wC(\P)$. Unlike Ezra's construction,
which has exponential dependency on the dimension, our has only $d^2$
dependency on the dimension.  The running time can be improved to
$O( (n+m +\eps^{-1}) \log^{d} \eps^{-1})$ if the net is constructed
for a given set of $m$ boxes, which is the case in our settings. Our
technique also gives an efficient algorithm for constructing an
$\eps$-net of size $O(\eps^{-1}\log\log \eps^{-1})$ for rectangles in
$d=2,3$, which is somewhat simpler than the one in
\cite{aes-ssena-10}.

Building on the earlier work
\cite{aes-ssena-10,e-nawea-10,pt-tlbse-11}, our work brings the key
insights of the results to the forefront. Our main new ingredient is
using poly-logarithmic different grids to ``capture'' the distribution
of the point sets. This idea appears in the work of Pach and Tardos
\cite{pt-tlbse-11} in the construction of the lower-bound, so it is
not surprising that it is useful in the construction itself (i.e.,
upper-bound).

\subparagraph{(B) A new \MWU-based $O(d^2\log\log\popt)$-approximation
   algorithm.}

We present two algorithms in \secref{mwu}. The first
is essentially the one by Agarwal-Pan that computes a hitting set of
the range space $(\VV,\Boxes)$ of size $O(d^2\popt \log\log
\popt)$. We show that it can be implemented in
$O(n^{(d+1)/2}\log ^3 n)$ expected time (\thmref{main:1}). To achieve
the desired running time, we need a data structure to perform all the
required operations on $\VV$ without ever explicitly constructing
it. In \secref{DS}, we present such a data structure, which exploits
the properties of the partitioning technique by Overmars and Yap
\cite{oy-nubkm-91}.

Our main result, \thmref{pierce:alg}, is a different and faster \MWU
algorithm tailored for boxes with expected running time
$O(n^{d/2}\log^{2d+3} n)$.  It exploits the duality between the
piercing-set problem and the independent-set problem, along with fast
approximation algorithms for these two problems. The basic idea is to
use a fast approximation algorithm to find a large set of independent
boxes among the light boxes identified by the \MWU algorithm in a
round. If the algorithm does not find such an independent set, then we
can use a simple piercing-set algorithm to compute a desired piercing
set. Otherwise we double the weight of the boxes in the independent
set. This idea dramatically reduces the number of rounds performed by
the algorithm from $O( \popt \polylog n)$ to $O( \polylog n)$.  The
idea of duality and approximation to speedup the \MWU is critical to
get a near linear running time in the plane. As far as we are aware,
the idea of quickly rounding the dual problem and using it to speedup
the primal algorithm, in the context of
set-cover/hitting-set/piercing-set for rectangles, was not used
before. This algorithm also relies on the data structure described in
\secref{DS}.

\subparagraph{(C) A Faster $O(2^{4d}\log\log\popt)$-approximation
   algorithm.}

In \secref{candidate}, we present an $O(n^{\log d+1})$-time
$O(2^{4d}\log\log\popt)$-approximation algorithm for the box-piercing
problem. Let $\X\subseteq\VV$ be the set of vertices that lie on the
boundary of the intersection of at most $1+\floor{\log d}$ boxes of
$\Boxes$. $|\X|=O(n^{\log d +1})$.  We show that for every point
$v \in \VV$, there exists a point $x \in \GroundSet$ such that
$|x \cap \Boxes| = |v \cap \Boxes|/2^{4d}$. Using
relative-approximations \cite{hs-rag-11} and the duality of piercing
and independent sets, we show that the range space $(\X,\Boxes)$ has a
hitting set of size $\kappa = O(2^{4d} \popt)$. We compute a
fractional hitting set of $(\X,\Boxes)$ using the algorithm by
\cite{ap-naghs-20} and then compute a weak-net as in the previous
algorithm. This leads to an $O(n^{\log d + 1})$-time algorithm with an
approximation ratio of $O(2^{4d} \log \log \popt)$.

\subparagraph{(D) Piercing, clustering, and multi-round algorithms.}

In \secref{iteralg}, we show that the natural algorithm of first
computing a piercing set $\P_0$ for a random sample of input boxes,
and then piercing the input boxes that are not pierced by $\P_0$ leads
to an efficient piercing algorithm. This algorithm can be extended to
run an arbitrary number of rounds. For clustering this idea was
described by Har-Peled \cite{h-cm-04} (but the idea is much older see
e.g. \cite{grs-cecal-01}). In particular, if the algorithm runs for
$\nRounds$ rounds, then it computes piercing sets $\nRounds$ times for
sets of boxes of size (roughly) $\popt^{1-1/\nRounds}n^{1/\nRounds}$
and returns a piercing set of size at most $\nRounds\cdot \kappa$,
where $\kappa$ is the size of the maximum piercing set returned in
each round. If we use the algorithm in \secref{mwu} we obtain an
$O( d^3 \log \log \popt)$ approximation algorithm that runs in
$n\log^{O(d)} n$ expected time, provided that
$\popt = O(n^{1/(d-1)})$.  See \thmref{main:2} and
\corref{near:linear:time}.  For $d=2$, a slightly more careful
analysis leads to the striking result that a piercing set of size
$O( \popt \log \log \popt)$ can be computed in $O( n \log \popt)$
expected running time provided that $\popt = O( n/ \log^{15} n)$, see
\corref{2d:linear:time}.  On the other hand, if we use the algorithm
in \secref{candidate}, we obtain $2^{O(d)} \log \log \popt$
approximation algorithm that runs in $n\log^{O(d)}n$ expected time,
provided that $\popt = O(n^{1/(2\ceil{\log d}+1)})$. See
\thmref{faster:bad:approx}, and \corref{near:linear:b:aprx}.

\subparagraph{(E) Dynamic Algorithms for Piercing.}  In
\secref{dynamic}, we consider the piercing problem for a set $\Boxes$
of boxes in $\Re^2$ in the \emph{dynamic} setting, i.e., at each step
a new box is inserted into or deleted from $\Boxes$.  Our goal is to
maintain an (approximately) optimal solution of the current set.  We
implement a dynamic version of the multi-round sampling based
algorithm, and attain a randomized Monte Carlo
$O(\log\log\popt)$-approximation algorithm with
$O(n^{1/2}\polylog n )$ amortized expected update time. The update
time improves to $O(n^{1/3}\polylog n )$ if $\Boxes$ is a set of
squares in $\Re^2$.  In principle, our approach extends to higher
dimensions but currently we face a few technical hurdles in its
efficient implementation (\secref{dynamic}).

\paragraph{Recent developments}

Since the initial publication of our results, Bhore and Chan
\cite{bc-fsdaa-25} have proposed a faster algorithm for the piercing
problem.
In particular, for a parameter $ \delta > 0 $, they present an
$O( n^{1+\delta})$-time algorithm that achieves an approximation
factor of
\begin{math}
   O\pth{\log \log \popt },
\end{math}
where the hidden constant factors in the approximation ratio depend
polynomially on $\delta^{-1}$ and exponentially on $d$ ($d^{O(d)}$ to
be precise).  In contrast, the approximation ratio of our algorithms
has a dependency of $d^2$ on the dimension.  Compared with
\thmref{main:2}, their algorithm is deterministic and faster when the
optimal piercing set size $\popt = \Omega(n^{1/{d+1}})$ (or
$\popt = \Omega(n^{1/(2{\ceil{\log_2 d} + 1})})$ if we allow
exponential dependency on $d$), while ours is randomized and has
expected time $O(n \log n)$ for $d=2$ and
$\popt= O( n /\log^{O(1)} n)$, see \corref{2d:linear:time}. For $d >2$
and $\popt =O( n^{1/(d-1)})$, the running time of our algorithm
increases to $O(n \log^{O(d)} n)$, see \corref{near:linear:time}.

Our initial motivation was to handle the case where the piercing
number is small, say, a constant or $O( \log n)$, which is often the
case in real-world applications, in which case our algorithm is faster
while having a better dependency on the dimension in the approximation
factor. They also present a faster dynamic algorithm with
$n^{O(\delta)}$ update time, which is significantly faster than ours,
but with the same worse dependency on the dimension. Their work
contains other problems for which they achieve significant speedup.

\paragraph{Paper organization}

The improved weak $\eps$-net construction is described in
\secref{weak-net}. The new \MWU algorithms are presented in
\secref{mwu}. The surprisingly faster $n^{O(\log d)}$ algorithm is
presented in \secref{candidate}. The multi-round speedup is presented
in \secref{iteralg}. The dynamic algorithm in the plane is presented
in \secref{dynamic}.
The data-structure enabling all of the above is presented in
\secref{DS}. Some  conclusions are drawn in \secref{conclusions}.

\section{Weak \TPDF{$\eps$}{eps}-net for boxes}
\seclab{weak-net}

Let $\P$ be a set of $n$ points in $\Re^d$,
$w:\P\rightarrow \Re_{\geq 0}$ a weight function, and $\eps \in (0,1)$
a parameter. We describe an algorithm for choosing a set
$\N \subset \Re^d$ of $O( \tfrac{1}{\eps} \log \log \tfrac{1}{\eps} )$
points, such that for any box $\bb$ in $\Re^d$ with
$w(\bb \cap \P) \geq \eps w(\P)$, we have $\bb \cap \N\neq
\emptyset$. For $d\leq 3$, the set $\N$ is an $\eps$-net, as
$\N \subseteq \P$. For $d > 3$ the set $\N$ is not necessarily a
subset of $\P$, and is thus a \emph{weak} $\eps$-net. Our overall
approach is inspired by previous work \cite{aes-ssena-10,e-nawea-10},
but it is somewhat simpler.  In particular, unlike previous work, the
constant in the size of the computed net is polynomial in $d$.  The
algorithm works in three stages:
\begin{compactenumI}
    \medskip%
    \item It reduces the problem to computing a piercing set of a set
    $\X$ on the integral grid, where $\X$ has
    $\nQ=O(\tfrac{1}{\eps}\log \tfrac{1}{\eps})$ points. Thus, it
    suffices to compute a piercing set of a family $\BBB$ of
    $O(\nQ^{2d})$ canonical boxes.

    \medskip%
    \item Next, it chooses a set of $O(\tfrac{1}{\eps})$ points that pierces
    all the heavy boxes of $\BBB$, except for a subset
    $\CBeps\subseteq \BBB$ of size
    $O( \tfrac{1}{\eps}\log^{O(d^2)}\tfrac{1}{\eps} )$.

    \medskip%
    \item Finally, the algorithm chooses a set of
    $O(\frac{d^2}{\eps}\log\log \frac{d}{\eps})$ points that stab all
    boxes of $\CBeps$.
\end{compactenumI}
\smallskip%
Merging the sets results in the desired piercing set.
The expected running time of the algorithm is
$O(n+\tfrac{1}{\eps}\log^{O(d^2)}\tfrac{1}{\eps})$. We show how the
algorithm can be adapted so that the piercing set is a subset of $\P$
for $d=2,3$. Finally, if the goal is to pierce a given set of $m$
$\eps$-heavy boxes, the running time can be improved to
$O(m \log^{d} \tfrac{1}{\eps})$.

\subsection{Reducing the number of points}
\seclab{reducing}

For the sake of completeness, we provide a simple direct proof of the
following known result.

\begin{lemma}
    \lemlab{dir}
    Let $\P$ be a set of $n$ points in $\Re^d$,
    $\wC:\P\xrightarrow{} \Re_{\geq 0}$ be a weight function,
    $\eps \in (0,1/2)$ be a parameter, and let $\Q \subseteq \P$ be a
    random subset of size
    $\nQ = c \tfrac{d}{\eps}\log\tfrac{d}{\eps}$, where $c$ is an
    appropriate constant independent of $d$. Sampling is done with
    repetition, and the probability of a point to be chosen is
    proportional to its weight. Then, for all boxes $\bb$ with
    $\wX{\bb \cap \P} \geq \eps \wX{\P}$, we have
    $|\bb \cap \Q | \geq (\eps/2)\nQ$, with probability at least
    $1- \eps^{O(d)}$.
\end{lemma}
\begin{proof}
    For simplicity, we prove this lemma for the unweighted case. It
    easily extends to the weighted case.  Let $H_i$ be a minimal set
    of hyperplanes orthogonal to the $i$\th axis, such that there are
    at most $\zeta = \eps n / (16 d)$ points of $\P$ between two
    consecutive hyperplanes of $H_i$. The region lying between two
    adjacent hyperplanes of $H_i$ is a \emphi{slab}. Set
    $\Grid=\bigcup\nolimits_{i=1}^d H_i$. Observe that
    $\cardin{\Grid}=O(d^2/\eps)$, where the constant hiding in the
    $O$-notation is independent of $d$.  A box that all its $2d$
    facets lie on the hyperplanes of $\Grid$ is a \emph{snapped} box.

    Given a box $\bb$ that contains at least $ \eps n$ points of $\P$,
    let $\bb_\Grid$ be the largest snapped box that is contained in
    $\bb$, i.e. $\bb_\Grid$ is obtained by shrinking $\bb$ so that its
    facets lie on hyperplanes of $\Grid$. Since
    $\bb \setminus \bb_\Grid$ can be covered by $2d$ slabs of $\Grid$,
    and each slab contains at most $\frac{\eps n}{16d}$ points of
    $\P$, it follows that
    \begin{equation*}
        0
        \leq
        \cardin{ \bb \cap \P } - \cardin{ \bb_\Grid \cap \P }
        \leq
        \frac{\eps}{8} n.
    \end{equation*}

    Let $\Boxes$ be the set of all $\frac{7\eps}{8}$-heavy snapped
    boxes. Observe that
    \begin{equation*}
        \cardin{\Boxes} \leq \prod_{i=1}^d \cardin{H_i}^2 \leq
        (n/\zeta)^{2d} \leq (32d/\eps)^{2d}.
    \end{equation*}
    Let $\N$ be a random sample of $\P$ with repetition of
    $\nQ = c \tfrac{d}{\eps} \log \tfrac{d}{\eps}$ points for some
    constant $c >1$. (Technically $\N$ is a multiset, but for
    simplicity assume it to be a set.)

    Consider any box $\bb_\Grid \in \Boxes$, and let
    $\alpha = \frac{|\bb_\Grid \cap \P|}{\eps n}\geq \frac{7}{8}$. Let
    $Y = \cardin{\N \cap \bb_\Grid}$.  The variable $Y$ is a sum of
    $\nQ$ random binary variables, each being $1$ with probability
    $\frac{|\bb_{\Grid}\cap\P|}{n}=\alpha\eps$. Therefore,
    \begin{align*}
      \mu = \Ex{Y} = \alpha\eps \nQ\geq\frac{7}{8}\eps \nQ
    \end{align*}
    Using Chernoff inequality\footnote{Specifically, for any
       $\delta \in (0,1)$,
       $\Prob{Y < (1-\delta)\Ex{Y}\bigr.} \leq \exp\pth{ -(\delta^2/2)
          \Ex{Y}\bigr.}$.}, we have
    \begin{align*}
      \gamma
      &=
        \Prob{ Y \leq \frac{\eps}{2} \nQ}
        \leq
        \Prob{ Y \leq (1-\tfrac{3}{7})\frac{7}{8}\eps \nQ}
      \leq%
      \Prob{ Y \leq (1-\tfrac{3}{7})\mu}
      \leq%
      \exp\pth{ -\mu \frac{(3/7)^2}{2}}
      \\&
      \leq
      \exp\pth{ -\frac{7}{8}\cdot \frac{(3/7)^2}{2} \cdot \eps \nQ}
      \leq %
      \exp\pth{ -c \frac{ d}{14} \log \frac{d}{\eps}}
      \leq
      (\eps/d)^{cd/14}.
    \end{align*}

    Hence,  the probability of a box in $\Boxes$
    containing less then $\frac{\eps}{2}\nQ$ points of $\N$ is at most
    \begin{equation*}
        |\Boxes|\gamma
        \leq%
        \Bigl(\frac{32d}{\eps}\Bigr)^{2d}
        \Bigl(\frac{\eps}{d}\Bigr)^{{cd}/{14}}
        <
        \eps^{O(d)},
    \end{equation*}
    for $c = O(1)$, as $\eps \leq 1/2$. Hence, with probability at
    least $1-\eps^{-O(d)}$, $\N$ has the desired property.~~
\end{proof}

\begin{remark}
    \lemref{dir} provides us with a stronger result than just an
    $\eps$-net, it shows that a sample of size
    $\Bigl.\Theta( \tfrac{d}{\eps} \log \tfrac{d}{\eps})$ has the
    property that any $\eps$-heavy box contains at least
    $\Theta( \log \tfrac{1}{\eps})$ points of the sample. Relative
    approximations \cite{hs-rag-11} provide even stronger guarantees,
    which are not needed here.
\end{remark}

Let $\Q\subseteq\P$ be a set of $\nQ$ points as in \lemref{dir}. For simplicity,
assume that all the points of $\Q$ have unique coordinates. In
particular, for a point $p \in \Q$, let
$G(p) = (i_1, i_2, \ldots, i_d)$ be the integral point, where $i_j$ is
the rank of the $j$\th coordinate of $p$ among the $\nQ$ values of the
points of $\Q$ in this coordinate. Consider the ``\si{gridified}''
point set
\begin{equation*}
    \X = G(\Q) = \Set{ G(p)}{p \in \Q}
    \subseteq%
    \IRX{\nQ}^d.
\end{equation*}

\begin{defn}
    A box $\bb$ is \emphi{canonical} if it is of the form
    $\prod_{i=1}^d [\ell_i,\ell_i']$, where
    $\ell_i, \ell_i' \in \IRX{\nQ} = \{1,\ldots, \nQ\}$,
    $\ell_i \leq \ell_i'$, and each of its facets contains one point
    of $\X$.  Let $\CBOXES$ be the set of all canonical boxes. Observe
    that $\cardin{\CBOXES} \leq\nQ^{2d}$.
\end{defn}

Clearly, it is sufficient to construct a piercing set for all the
canonical boxes that are $(\eps/2)$-heavy with respect to $\X$.

\subsection{Reducing the number of canonical boxes}
\seclab{heavyCanBox}

We next construct a piercing set $\N_1$, of size $O(\tfrac{d}{\eps})$,
such that the residual set $\CBeps \subseteq \CBOXES$, of all boxes
not stabbed by $\N_1$, is of size at most
$\tfrac{1}{\eps}\log^{O(1)}\tfrac{1}{\eps}$ (the constant hiding in
the $O$-notation depends on $d$).  Assume that
\begin{equation}
    \nQ=2^{\lgEps} = \alpha \frac{d}{\eps} \log \frac{d}{\eps},%
    \qquad\text{for some integer}
    \qquad%
    \lgEps = \log \nQ
    =%
    O\Bigl(  \log  \frac{d}{\eps}  \Bigr),
    \eqlab{lg:eps}
\end{equation}
where $\alpha \geq 1$ is some appropriate constant.

\begin{defn}
    \deflab{canonical:boxes}%
    Let $\IRY{-1}{\lgEps} = \{-1,0,\ldots, \lgEps\}$, see
    \Eqref{lg:eps}. Given
    $\ii = (i_1, \ldots, i_d) \in \IRY{-1}{\lgEps\,}^d$, we associate
    with it the box
    $\bb_\ii = [0,2^{i_1}] \times [0, 2^{i_2}] \times \ldots \times
    [0,2^{i_d}]$. Such a box $\bb_\ii$ naturally tiles the bounding
    cube $[0,\nQ]^d \supset \X$ into a grid, and let $\Grid_\ii$
    denote this grid. The associated set of grids is
    \begin{equation*}
        \GridSet=\Set{\Grid_\ii}{ \ii\in \smash{\IRY{-1}{\lgEps}^d}},
    \end{equation*}
    and $\BBGS$ be the set of all boxes that appear as grid cells in
    any grid of $\GridSet$.  Observe that
    \begin{equation*}
        \nG
        =%
        |\GridSet| = (\lgEps+1)^d = O\bigl( \log^d \tfrac{d}{\eps}
        \bigr).
    \end{equation*}
    For a box $\bb = \prod_{i=1}^d [\ell_i, \ell_i']$ its
    \emphi{center} is the point
    $\cenX{\bb} = \bigl( (\ell_1+\ell_1')/2, \ldots,
    (\ell_d+\ell_d')/2 \bigr)$. For a canonical box $\bb$, a box
    $\bbC \in \BBGS$, such that $\bb \subseteq \bbC$ and
    $\cenX{\bbC} \in \bb$ is a \emphi{frame} of $\bb$.
\end{defn}

\begin{lemma}
    \lemlab{frame}%
    For any canonical box $\bb \in \CBOXES$ that is $\eps/2$-heavy
    with respect to $\X$, there exists a box $\bbD \in \BBGS$ that is
    a frame of $\bb$.
\end{lemma}
\begin{proof}
    For a box $\bbA$, a \emphw{halving hyperplane} is a hyperplane
    passing through the center point $\cenX{\bbA}$ that is orthogonal
    to one of the axes.

    Let $\bbA_0 = [0,\nQ]^d$ be the initial cell, with
    $\bb\subseteq \bbA_0$. For $i >0$, in the $i$\th iteration, if
    there is a halving plane $h_i$ of $\bbA_{i-1}$ that does not
    intersect $\bb$, then we split $\bbA_{i-1}$ by $h_i$ into two
    equal boxes, and set $\bbA_i$ to be the half that contains
    $\bb$. We continue this process until all halving hyperplanes of
    $\bbA_i$ intersect $\bb$. Clearly, the resulting box $\bbD$
    contains $\bb$, $\bbD$ is a box of $\BBGS$, and the center point
    of $\bbD$ is contained in $\bb$, since all halving hyperplanes of
    $\bbD$ intersect $\bb$.
\end{proof}

A cell $\bbA \in \BBGS$ is \emphi{massive} if it contains at least
\begin{equation*}
    \nMass = (3 \lgEps)^{d+2}
    =
    2^{c_2 d} \log^{d+2} \frac{d}{\eps}
\end{equation*}
points of $\X$ in its interior, where $c_2$ is some constant
independent of $d$, see \Eqref{lg:eps}.

\begin{lemma}
    \lemlab{massive-bound}%
    The set $\BBGS$ contains at most
    \begin{math}
        c'd / \eps
    \end{math}
    massive boxes, where $c'$ is a constant independent of $d$.
\end{lemma}
\begin{proof}
    Fix a grid $\Grid$ out of the $\gamma = (2+\lgEps)^d$ grids of
    $\GridSet$ (see \defref{canonical:boxes}). Such a grid can contain
    at most
    \begin{equation*}
        \gamma' =
        \frac{\cardin{\X}}{ \nMass}%
        =%
        \frac{\nQ}{\nMass}%
        =%
        \frac{\nQ}{(3 \lgEps)^{d+2}}
    \end{equation*}
    massive cells. Since there are $\gamma$ different grids, it
    follows that the total number of massive cells is at most
    \begin{math}
        \gamma \gamma' \leq \nQ / (3\lgEps)^2
        \leq
        c' d/\eps,
    \end{math}
    where $c'$ is a constant that is independent of $d$, see
    \Eqref{lg:eps}.
\end{proof}

\begin{observation}
    All canonical boxes, whose frames are massive, are stabbed by the
    centers of these massive boxes.
\end{observation}

\begin{lemma}
    \lemlab{b:b:a:size}%
    Let $\CBeps\subseteq \CBOXES$ be the set of canonical boxes that
    are $(\eps/2)$-heavy with respect to $\X$, but their frames are
    not massive. We have
    \begin{equation*}
        |\CBeps|
        \leq %
        \frac{2^{c_3 d^2}}{\eps} \log^{4d^2} \frac{d}{\eps},
    \end{equation*}
    where $c_3$ is a constant independent of $d$.
\end{lemma}
\begin{proof}
    Fix a grid $\Grid_\ii$ of $\GridSet$.  A cell
    $\cell \in \Grid_\ii$ that is a frame of a box $\bb \in \CBeps$,
    must contain at least one point of $\X$. Hence,
    \begin{math}
        \sum_{\cell \in \Grid_\ii} \cardin{\cell \cap \X} \leq 2^d
        \nQ,
    \end{math}
    since a point of $\X$ can participate in at most $2^d$ cells of
    $\Grid_\ii$.  On the other hand, such a cell $\cell$ can host
    inside it at most $\cardin{\cell \cap \X}^{2d}$ canonical
    boxes. Thus, the total number of boxes of $\CBeps$ that have
    frames in $\Grid_\ii$, that are not massive, is at most
    $\Bigl.\sum_{\cell \in \Grid_\ii: \cardin{\cell \cap \X} < \nMass}
    \cardin{\cell \cap \X}^{2d} $. Observing that this quantity is
    maximized by the frames being as heavy as possible (that is,
    $\nMass -1$) and summing the bound over all grids of $\GridSet$,
    we obtain
    \begin{align*}
      |\CBeps|
      &\leq%
        \sum_{\Grid_\ii \in \GridSet }
        \;
        \sum_{\cell \in \Grid_\ii: \cardin{\cell \cap \X} < \nMass}
        \cardin{\cell \cap \X}^{2d}
        \leq%
        (\lgEps+1)^d\cdot\frac{2^d\nQ}{\nMass}\cdot\nMass^{2d}
      \\&%
      \leq%
      (\lgEps+1)^d\cdot{2^d\cdot\nQ}\cdot {(3\lgEps)}^{(d+2)\cdot(2d-1)}
      \leq%
      {2^d}\cdot {(3\lgEps)}^{2d^2+4d-2}\cdot2^{\lgEps}
      \\&
      \leq%
      \frac{2^{c_3d^2}}{\eps}\log^{4d^2}\frac{d}{\eps},
    \end{align*}
    where $c_3>0$ is a constant independent of $d$.
\end{proof}

\begin{remark}
    \remlab{big:enough}%
    If $|\CBeps| =O(\tfrac{1}{\eps}\log\log \tfrac{1}{\eps})$, then we
    can choose a point inside each such box into our piercing set, and
    we are done. Thus, from this point on, assume $|\CBeps| > 4/\eps$.
\end{remark}

\subsection{Piercing boxes with light frames }

Finally, we choose a set of
$O(\tfrac{1}{\eps}\log\log \tfrac{1}{\eps})$ points that pierce all
boxes of $\CBeps$ as follows. For
\begin{equation*}
    \nu = |\CBeps|
    \leq \frac{2^{c_3 d^2}}{\eps} \log^{4d^2} \frac{d}{\eps}
    \qquad\text{and}\qquad%
    m = \frac{4}{\eps}\ln (\eps\nu ),
\end{equation*}
let $\R_1$ be a random sample $\X$ of size $m$ (say, with
repetitions). Let $\BBA_0\subseteq \CBeps$ be the set of boxes that
are not pierced by $\R_1$. If $|\BBA_0|\leq 4/\eps$, we choose one
point from each box of $\BBA_0$. Let $\R_2$ be the resulting set of
points. We return $\R_1 \cup \R_2$ as the piercing set of $\BBA$. If
$|\BBA_0| >4/\eps$, we discard $\R_1$ and repeat the above steps. The
following lemma shows that
$\Prob{\bigl.\smash{|\BBA_0|} > 4/\eps} < \frac{1}{4}$, so the
algorithm stops in (expected) constant number of rounds.

\begin{lemma}
    \lemlab{leftover}%
    We have $\Prob{\bigl.\smash{|\BBA_0|} > 4/\eps} < \frac{1}{4}$.
\end{lemma}
\begin{proof}
    Let $\nu = \cardin{\smash{\BBA}}$.  By \lemref{b:b:a:size} and
    \remref{big:enough}, we have
    \begin{equation*}
        1 \leq \eps \nu
        \leq%
        2^{ c_3 d^2}  \log^{4d^2} \frac{d}{\eps}.
    \end{equation*}
    For a specific box $\bb \in \BBA$ (which is $\eps/2$-heavy by
    definition), we have
    \begin{align*}
      \tau
      =
      \Prob{\bigl.\bb\cap \R_1= \emptyset}
      &=%
        \pth{ 1 - \frac{\cardin{\bb \cap \X}}{\cardin{\X}}}^m
        \leq
        \pth{ 1 - \frac{\eps}{2}}^m
        \leq
        \exp \pth{  - \frac{\eps}{2} m}
        \leq%
        \frac{1}{\eps^2 \nu^2},
    \end{align*}
    as
    $ \tfrac{\eps}{2} m = \tfrac{\eps}{2} \cdot \tfrac{4}{\eps}\ln
    (\eps\nu ) = 2 \ln (\eps \nu)$, Thus, by linearity of
    expectations, the expected number of boxes of $\BBA\Bigr.$ that
    failed to be hit by $\R_1$ is
    \begin{equation*}
        \Ex{\cardin{\smash{\BBA_0}} \bigr.}%
        \leq%
        \cardin{\bigl.\smash{\BBA}} \tau
        =
        \frac{\nu}{\eps^2 \nu^2}
        =%
        \frac{1}{\eps^2 \nu}
        \leq%
        \frac{1}{\eps}.
    \end{equation*} By Markov's inequality,
    $\Prob{\bigl. \smash{|\BBA_0|} > \tfrac{4}{\eps}} < \frac{1}{4}$.
\end{proof}

As for the size of the random sample, we have
\begin{equation*}
    m
    =%
    |\R_1|
    =%
    \frac{4}{\eps} \log (\eps\nu)%
    =%
    \frac{4}{\eps}
    \log \pth{%
       2^{ c_3 d^2}  \log^{4d^2} \frac{d}{\eps}
    }
    =%
    O\pth{ \frac{d^2}{\eps} \log \log \frac{d}{\eps} }.
\end{equation*}
Hence
$\Bigl.|\R_1 \cup \R_2| \leq c_4d^2\tfrac{1}{\eps}\log\log
\tfrac{d}{\eps}$, where $c_4$ is a constant.

By preprocessing $\R_1$ using range-searching data structure, in
\begin{equation*}
    O(|\R_1|\log^{d-1}|\R_1|)=O(\tfrac{1}{\eps}\log^d\tfrac{1}{\eps})
\end{equation*}
time, the set $\BBA_0$ can be computed in
$\Bigl.O(\nu\log^{d-1}
\nu)=\tfrac{1}{\eps}\log^{O(d^2)}\tfrac{1}{\eps}$ time. Finally, the
set $\R_2$ can be computed in $O(\tfrac{1}{\eps})$ time. Hence, we
obtain the following:

\begin{lemma}
    \lemlab{net:light}%
    For some constant $c_4 > 1$, a piercing set of $\BBA$ of size at
    most $c_4d^2\tfrac{1}{\eps}\log\log \tfrac{d}{\eps}$ can be
    computed in
    \begin{math}
        \tfrac{1}{\eps} \log^{O(d^2)}\tfrac{1}{\eps}
    \end{math}
    expected time.
\end{lemma}

\subsection{Putting it together}

We first choose a random subset $\Q\subseteq\P$ of $\nQ$ points and
map them to a set $\X$ of grid points, as described in
\secref{heavyCanBox}. Next, we construct the family of grids
$\GridSet$.  For a grid $\Grid_\ii$ of $\GridSet$ we can identify all
the massive cells in it (and all the cells that are not massive, but
are $\eps/2$-heavy) in $O(\cardin{\X})$ time (using hashing to
``store'' the points of $\X$ in the grid $\Grid_\ii$.  Thus, we can
compute the set $\BBM$ of all massive cells in the grids of $\GridSet$
in
\begin{equation*}
    O(\cardin{\X}
    \lgEps^d
    )
    =%
    O( \nQ \log^d \nQ)
\end{equation*}
time, see \Eqref{lg:eps}. Let $\N_1$ be the set of centers of cells in
$\BBM$. By \lemref{massive-bound}, we have $|\N_1|\leq
c'd/\eps$. Finally, for each (light) grid cell
$\cell\in \BBGS\setminus \BBM$, we compute the set of canonical boxes
that are contained in $\cell$ and are $(\eps/2)$-heavy with respect to
$\X$. Let $\BBA$ be the resulting set of boxes. By
\lemref{b:b:a:size},
$|\BBA|\leq {2^{c_3 d^2} }{\tfrac{1}{\eps}} \log^{4d^2}
\tfrac{d}{\eps}$, and it can be computed in
$O(|\BBGS|\log^{d-1}\nQ+|\BBA|)=O(\tfrac{1}{\eps}\log^{O(d^2)}\tfrac{1}{\eps})$
time. Finally, we construct a piercing set $\N_2$ of $\BBA$ of size at
most $c_4d^2\tfrac{1}{\eps}\log\log \tfrac{d}{\eps} $, using
\lemref{net:light}. The set $\N_1\cup\N_2$ pierces all canonical boxes
that are $\tfrac{\eps}{2}$-heavy with respect to $\X$.

The set $\N_1 \cup \N_2$ can readily be ``translated'' into a piercing
set $\N$ of the same size for the original set of points, which
pierces all the $\eps$-heavy boxes of $\P$.  Combining the above, we
get the following.

\begin{theorem}
    \thmlab{main:w:net}%
    Let $\P$ be a set of $n$ points in $\Re^d$,
    $w:\P\rightarrow\Re_{\geq 0}$ a weight function, and
    $\eps\in (0,1)$ a parameter. A set $\N$ of
    $c_4 \tfrac{d^2}{\eps} \log \log \tfrac{d}{\eps}$ points can be
    computed in $O(n+\tfrac{1}{\eps}\log^{O(d^2)}\tfrac{1}{\eps})$
    expected time, such that every $\eps$-heavy box $\bb$ (with
    respect to $\P$) contains at least one point of $\N$.
\end{theorem}

\subsection{An improved algorithm for weak \TPDF{$\eps$}{eps}-net}
\seclab{alt:algorithm}

\begin{lemma}
    \lemlab{net}%
    Given a set $\P$ of $n$ points in $\Re^d$, a parameter
    $\eps \in (0,1)$, and a set $\Boxes$ of $m$ boxes in $\Re^d$ that
    all contain at least $\eps n$ points of $\P$. Then, a weak
    $\eps$-net $\N$ of size
    $O(\tfrac{1}{\eps} \log \log \tfrac{1}{\eps} )$ that pierces every
    box $\bb \in \Boxes$ can be computed in
    $O\bigl((m + \tfrac{1}{\eps} \log \tfrac{1}{\eps} ) \log^{d-1}
    \tfrac{1}{\eps} \bigr)$ expected time.
\end{lemma}
\begin{proof}
    If $m \leq 1/\eps$, then we pick a point from each box of
    $\Boxes$, and we are done.  Otherwise, we roughly follow the above
    construction. We first sample a set $\Q$ of size
    $\nQ = O( \tfrac{1}{\eps} \log \tfrac{1}{\eps} )$, and compute
    $\X$ from it, as done above. This takes
    $O( \tfrac{1}{\eps} \log^2 \tfrac{1}{\eps} )$ time. Next, we map
    every box of $\Boxes$ to the appropriate box in $[0,\nQ]^d$. We
    shrink it so that its corners all have integral coordinates. This
    takes $O( m \log \tfrac{1}{\eps})$ time. Let $\Boxes'$ be the
    resulting set of boxes.

    We compute for each box of $\Boxes'$ its frame using the
    construction algorithm of \lemref{frame}. This can be done in
    $O( m d )$ time by using \LCA queries in the appropriate tree.

    Next, we choose a sample $\N_1$ from $\X$
    of size $O( \tfrac{1}{\eps} \log \log \tfrac{1}{\eps})$. Using
    range tree built for $\N_1$, we compute all the boxes in $\Boxes'$
    not being hit by $\N_1$. Let $\Boxes''$ be this set of boxes. This
    takes $O( (m+\tfrac{1}{\eps}) \log^{d-1} \tfrac{1}{\eps})$ time
    \cite{a-rs-04}.

    For each unstabbed box $\bb \in \Boxes''$, let $\bbA$ be its
    frame. If $\bbA$ is massive, we add its center to the weak
    $\eps$-net. Let $\N_2$ be the resulting set of new stabbing
    points.

    The argument of \lemref{net:light} then implies that, in
    expectation, the number of remaining unstabbed boxes in $\Boxes''$
    (with non-massive frame) is $O(1/\eps)$. We pick a point of $\X$
    from each one these boxes, and let $\N_3$ be the set of these
    points.  Clearly, $\N_1 \cup \N_2 \cup \N_3$ is the desired
    piercing set, of the desired set as it is just a slightly
    different implementation of the process described above. If this
    set is too large, we restart the algorithm till success.  Overall,
    the total running time is dominated by orthogonal range searching
    queries, that take
    $O\bigl( (m + \tfrac{1}{\eps}\log\tfrac{1}{\eps} )
    \log^{d-1}\tfrac{1}{\eps} )$.
\end{proof}

\subsection{An \TPDF{$\eps$}{eps}-net in two and three dimensions}

Here, we show how one can convert the construction of
\thmref{main:w:net} so that it generates an $\eps$-net (instead of a
\emph{weak} $\eps$-net) in two and three dimensions.  The construction
uses points that are outside the set $\P$ only when adding the centers
of the massive boxes to the piercing set (see \lemref{massive-bound}).
We describe only the 3d case, as the 2d case is implied by it.

\begin{defn}
    A \emphi{crate} in $\Re^3$ is a box $\bb$ that one of its corners
    is the center of a massive box $\bbA \in \BBGS$,
    $\bb \subseteq \bbA$, $\bb$ does not contain any point of $\P$ in
    its interior, and it is maximal (under containment).  A
    \emphi{$k$-crate} is a box that is a crate after deleting $k$
    points of $\P$ that lie in its interior.
\end{defn}

\begin{lemma}
    \lemlab{crates}%
    Let $\P$ be a set of $n$ points in $\Re^3$, and let
    $\bbA \in \BBGS$ be a massive box. Then, the number of crates
    contained in $\bbA$ with a corner at $\cenX{\bbA}$ is
    $O( \cardin{ \P \cap \bbA })$.

    Furthermore, the number of $k$-crates contained in $\bbA$ is at
    most $O(k^2\cardin{ \P \cap \bbA })$.
\end{lemma}

\begin{proof}
    Assume the center of $\bbA$ is in the origin, and without limiting
    generality only count the crates with the diametrical corner being
    in the positive octant (the bound applies symmetrically to the
    remaining seven octants).

    So, let $\Q$ be the set of $\nu$ points of $\P \cap \bbA$ in the
    positive octant, and sort them by increasing $z$ coordinate, with
    $\Q_i = \{ q_1, q_2, \ldots, q_i \}$, for $i=1,\ldots, \nu$, being
    the prefix set. Consider a crate $\bbC$ in the positive octant,
    that has the points $q_i, q_j$ on its vertical faces (not adjacent
    to the origin), and $q_k$ on its top face. Observe that $k > i$
    and $k > j$, and assume that $i < j <k$. We charge the crate
    $\bbC$ to $j$. Observe that $j$ can be charged at most twice --
    indeed, consider the projection $\Q_j$ to the $xy$-plane, denoted
    by $\Q_j'$. Clearly, the projection of $\bbC$ into the $xy$-plane
    is a maximal box with a corner at the origin, no points of $\Q_j'$
    in its interior, and $q_i'$ and $q_j'$ on its two edges not
    adjacent to the origin, see \figref{stairway}.

    Note, that there could be at most two such boxes having $q_j'$ on
    its boundary, which readily implies $q_j'$ get charged at most
    twice. This readily implies the first part of the claim.

    \begin{figure}[h]
        \centering%
        \includegraphics{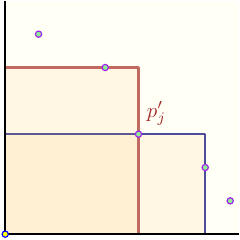}
        \caption{The two crates charged to $p_j$, and their projection
           to the $xy$-axis.}
        \figlab{stairway}
    \end{figure}

    The second implication is readily implied by the Clarkson-Shor
    technique \cite{cs-arscg-89} as every crate is defined by three
    points.
\end{proof}

\begin{lemma}
    In $\Re^3$, the center points of massive boxes used by the
    piercing set of \thmref{main:w:net}, can be replaced by a subset
    of $O( \tfrac{1}{\eps}\log \log \tfrac{1}{\eps} )$ points of $\P$,
    so that the resulting set is an $\eps$-net, of size
    $O( \tfrac{1}{\eps}\log \log \tfrac{1}{\eps} )$, for all the
    $\eps$-heavy boxes
\end{lemma}

\begin{proof}
    Let $\BBM$ be the set of all massive cells in the grids of
    $\GridSet$, By \lemref{dir}, all the $\eps$-heavy boxes of $\P$
    contain at least $(\eps/2)\nQ$ points of $\X$.  Such an
    $\eps$-heavy box $\bb$ contained in a massive cell $\bbC$, that is
    stabbed by its center $\cenX{\bbC}$, contains a crate
    $\bbA \subseteq \bb$ that has at least (and thus, exactly)
    $k = (\eps/16)\nQ = O( \log \tfrac{1}{\eps})$ points of
    $\X$. Thus, a massive box containing $n'$ points, given rise to
    $O( k^2 n')$ crates, by \lemref{crates}. Summing over all cells in
    $\BBM$, we have that the total number of $k$-crates (in massive
    cells) is at most
    \begin{equation*}
        \sum_{\cell \in \BBM }
        O(k^2 \cardin{\cell \cap \X})
        =
        O( k^2 \cardin{\GridSet} \nQ )
        =
        O( \tfrac{1}{\eps}\log^6 \tfrac{1}{\eps} ).
    \end{equation*}
    Now, arguing as above, a random sample of size
    $O( \tfrac{1}{\eps} \log \log \tfrac{1}{\eps} )$ from $\X$ would
    stab all the heavy crates under consideration except for
    $O(1/\eps)$ of them. Adding a point for each of these unstabbed
    creates, then complete the construction of the $\eps$-net.
\end{proof}

We thus get the following extension of \thmref{main:w:net}.

\begin{theorem}
    \thmlab{main:strong:net}%
    Let $\P$ be a set of $n$ points in $\Re^d$, for $d=2$ or $d=3$,
    $w:\P\rightarrow\Re_{\geq 0}$ a weight function, and
    $\eps\in (0,1)$ a parameter. A set $\N \subseteq \P$ of
    $O( \tfrac{1}{\eps} \log \log \tfrac{d}{\eps} )$ points can be
    computed, in $O(n+\tfrac{1}{\eps}\log^{O(1)}\tfrac{1}{\eps})$
    expected time, such that every $\eps$-heavy box $\bb$ (with
    respect to $\P$) contains at least one point of $\N$.
\end{theorem}

\section{Piercing via multiplicative weight update (\MWU)}
\seclab{mwu}

The input is a set $\Boxes$ of $n$ boxes in $\Re^d$. Formally, a
\emphi{box} is an axis-aligned closed set of the form
$\bb = [\ell_1,\ell_1'] \times \cdots \times [\ell_d,\ell_d']$.  A
\emphi{piercing set} $\Q \subseteq \Re^d$, such that for all
$\bb \in \Boxes$, the set $\bb \cap \Q$ is not empty (i.e., all the
boxes of $\Boxes$ are stabbed by some point of $\Q$). The minimum
cardinality of a piercing set of $\Boxes$ is the \emphi{piercing
   number} of $\Boxes$, denoted by $\popt$.  The task at hand is to
compute a piercing set of minimum size, thus providing an efficient
approximation to $\popt$.

\bigskip

We describe two multiplicative-weight-update (\MWU) based
$O(\log\log \popt)$-approximation algorithms for computing a piercing
set for $\Boxes$, where $\popt$ is the optimal piercing number. Let
$\VV=\VX{\Boxes}$ denote the set of vertices in the arrangement
$\ArrX{\Boxes}$. As already mentioned, there is always an optimal
piercing set $\N \subseteq \VV$ of $\Boxes$.  The basic idea is to
reweight the points of $\VV$ such that all boxes become
\textit{heavy}.  Specifically, we compute a weight function
$\wC:\VV \xrightarrow{} \ZZ^{+}$, such that for any box
$\bb \in \Boxes$, we have
\begin{math}
    \wX{\bb \cap \VV}%
    =%
    \sum_{ p \in \bb \cap \VV} \wX{p}%
    \geq%
    \tfrac{\wX{\VV}}{{2e \popt }} \Bigr. .
\end{math}
We then compute a weak $\tfrac{1}{2e \popt }$-net $\N$ of the range
space $(\VV,\Boxes)$ with respect to the above weight function. By
definition $\N$ is a piercing set of $\Boxes$. The two algorithms
differ in how the weights are updated. We next specify the data
structure used by both algorithms. The implementation details of the
data structure are in \secref{DS}.

\begin{figure}
    \centering
    \begin{tabular}{|c|c|c|c|l|}
      \hline
      Dim
      &
        preprocessing
      &
        Std. operations
      &
        \sample
      & ref\\
      \hline
      $d=1$
      &
        $O(n \log n)$
      &
        $O( \log n)$
      &
        $O( \log n)\Bigr.$
      &
        \lemref{d:s:1:dim}%
      \\
      \hline
      $d=2$
      &
        $O(n^{3/2} \log n)$
      &
        $O( \sqrt{n}\log n )$
      &
        $O( \log n)\Bigr.$
      &
        \lemref{data-structure}%
      \\
      \hline
      $d>2$
      &
        $O(n^{(d+1)/2} \log n)$
      &
        $O( n^{(d-1)/2}\log n )\Bigr.$
      &
        $O( \log n)$
      &
        \lemref{data-structure}%
      \\
      \hline
    \end{tabular}
    \caption{%
       We describe a data-structure for maintaining (implicitly) the
       vertices of an arrangement of boxes, under operations \weight,
       \double, \halve, \Insert, \Delete, see \defref{arr:ds1}.}
    \figlab{d:s:p}
\end{figure}

For a multi-set $\S \subseteq \Boxes$ of boxes, and a
point $p\in \Re^d$, let $\rstY{\S}{p} = \Set{ \bb \in \S}{ p \in \bb}$
be the multi-set of all boxes in $\S$ containing $p$, let
\begin{equation}
    \hwY{p}{\S} = 2^{\nrstY{\S}{p}}
    \eqlab{doubling:w}%
\end{equation}
be the \emphi{doubling weight} of $p$. For a finite set of points $X\subseteq \Re^d$,
let $\hwY{X}{\S} = \sum\nolimits_{x \in X} \hwY{p}{\S}$.

\begin{defn}[Implicit arrangement data-structure]
    \deflab{arr:ds1}%
    \deflab{arr:ds}%
    Let $\Boxes$ be a set of axis-aligned boxes in $\Re^d$ (known in
    advance). Let $\AB\subseteq \Boxes$ be a set of \emphi{active}
    boxes (initially empty) and let $\VX{\AB}$ be the set of vertices
    of the arrangement $\ArrX{\AB}$. Let $\S\subseteq \Boxes$ be a
    multi-set of \emphi{update} boxes (initially empty). The set $\S$
    induces a weight function on the vertices of $\VV(\AB)$ using
    \Eqref{doubling:w}. Here, we require a data-structure that
    supports the following operations:
    \begin{compactenumI}
        \smallskip%
        \item $\weight (\bb)$: given a box $\bb$, compute $\hwY{\VX{\AB}\cap \bb}{\S}$.

        \smallskip%
        \item $\double(\bb)$: given a box $\bb \in \Boxes$, adds a
        copy of $\bb$ to $\S$.

        \smallskip%
        \item $\halve(\bb)$: given a box $\bb \in \Boxes$, removes a
        copy of $\bb$ from $\S$.

        \smallskip%
        \item $\sample$: returns a random point $p \in \VV(\AB)$, with
        probability ${\hwY{p}{\S}} / {\hwY{\bigl.\VX{\AB}}{\S}}$.

        \smallskip%
        \item $\Insert (\bb)$: inserts $\bb$ into the set of active
        boxes $\AB$.

        \smallskip%
        \item $\Delete (\bb)$: removes $\bb$ from the set  of
        active boxes $\AB$.

    \end{compactenumI}
\end{defn}

\figref{d:s:p} summarizes the performance of the data structure. In
particular, it can be constructed in $O(n^{(d+1)/2}\log n)$ time, each
of the above operations can be supported in $O(n^{(d-1)/2}\log n)$
time per operation, except for $\sample$, which takes only $O(\log n)$
time, see \lemref{data-structure}. Note that the $\halve$ and
$\Delete$ operations are not used by the basic algorithm presented
next.

\subsection{Basic \MWU algorithm}
\seclab{mwu:basic}

The algorithm is a variant of the algorithm of Agarwal and Pan
\cite{ap-naghs-20}. The algorithm performs an exponential search on a
value $k$, and it stops when $k \geq \popt$ and $k\leq
2\popt$. Specifically, in the $i$\th stage, $k$ is set to $2^i$. We
next describe such a stage for a fixed value of $k$ (and $i$).

In the beginning of the stage, $\wX{p}=1$ for all $p\in \VV$.  Let
$\eps = \frac{2}{3k}$. A box $\bb$ is \emphi{$\eps$-light} if
$\wX{\VV \cap \bb} < \eps \wX{\VV}$, and \emphi{$\eps$-heavy}
otherwise. The algorithm proceeds in rounds. In each round, it scans
all the boxes in $\Boxes$ trying to find an $\eps$-light box.  If an
$\eps$-light box $\bb$ is found, the algorithm performs a
\emphi{doubling} operation on its points. That is, it doubles the
weight of each of the points in $\VV \cap \bb$. The algorithm performs
the doubling on $\bb$ repeatedly until $\bb$ becomes $\eps$-heavy (in
relation to the updated weight $\wX{\VV}$).  The algorithm then
resumes the scan for $\eps$-light boxes (over the remaining
boxes). Importantly, the algorithm never revisits a box during a round
(thus a box that becomes heavy during a round might become light later
in the same round).  Let $\ell = \floor{ 1/\eps} \leq 2k$. If $\ell$
doubling operations are performed in a round, the algorithm aborts the
round, and proceeds to the next round.  If a round was completed
without $\ell$ weight-doubling being performed, the algorithm computes
a weak $(\eps/e)$-net $\N$ of $(\VV,\Boxes)$ using the algorithm
described below in \secref{weak-net} (specifically, \lemref{net}).

Finally, if the number of rounds exceeds $\tau = c\ln( |\VV| / k ) $
at any stage, where $c>0$ is a suitably large constant, the guess for
$k$ is too small, and the algorithm doubles the value of $k$, and
continues to the next stage, till success.

In order to implement the above algorithm, we use the implicit
arrangement data structure from \defref{arr:ds1}. At the beginning of
each stage, we build an instance of the data structure on $\Boxes$ and
add all the boxes to the active set using the $\Insert$ operation. The
doubling step with a box $\bb$ is performed using $\double(\bb)$ and
the $\eps$-lightness of a box is tested using the $\weight(\bb)$
operation.
\paragraph{Analysis}
\seclab{analysis}

The following lemma  proves the correctness of the
algorithm.

\begin{lemma}[\tcite{ap-naghs-20}]
    \lemlab{eps:net:2d:log:log}%
    If $k\in [{\popt}/{2},\popt ]$, then the algorithm returns a
    piercing set for $\Boxes$, of size $O( \popt \log \log \popt)$,
    where $\popt =\PNX{\Boxes}$ is the piercing number of $\Boxes$.
\end{lemma}
\begin{proof}
    This claim is well known \cite[Lemma 3.1]{ap-naghs-20}, but we
    include a proof for the sake of completeness.  Initially, the
    weight $W_0$ of all the vertices of $\VV$ is
    $\mm = |\VV| = O(n^d)$. Let $W_i$ be the weight of $\VV$ after the
    $i$\th doubling operation, and observe that
    \begin{equation*}
        \popt 2^{\floor{i/\popt}} \leq W_i \leq (1+\eps)^i W_0
        \leq%
        \mm \exp\pth{ \eps i}.
    \end{equation*}
    The lower bound follows as every doubling operation must double
    the weight of one of the $\popt$ points in the optimal piercing
    set (and to minimize this quantity, this happens in a round robin
    fashion). The upper bound follows readily from the
    $\eps$-lightness of the box being doubled.  Taking $i=t \popt$,
    and taking the log of both sides (in base $2$), we have
    \begin{equation*}
        \log \popt + t
        \leq
        \log \mm +  \eps t \popt \log e
        \leq
        \log \mm +
        \frac{2}{3k} t \popt \cdot 1.45
        \leq%
        \log \mm + 0.97  t \frac{\popt}{k} .
    \end{equation*}
    Assuming $k \geq \popt$, this readily implies that
    $t = O( \log \tfrac\mm{k})$. Namely, the algorithm performs at
    most $\tau = O( k \log \tfrac{\mm}{k})$ doubling operations in a
    stage if $k \geq \popt$.

    Since every round performs exactly
    $\ell = \floor{1/\eps} = \floor{3k/2} \geq k$ doubling operations
    in each round, except the last one, it follows that the algorithm
    performs at most
    $\tau / \ell \leq \tau/ k = O( \log \tfrac{\mm}{k})$ rounds (in a
    single stage).

    So consider the last round -- all the boxes where scanned
    ``successfully'' during this round. Assume the weight of $\VV$ at
    the start of the round was $W$. At the end of the round, the total
    weight of $\VV$ is at most
    $(1+\eps)^\ell W \leq \exp( \eps \floor{ 1/\eps }) W \leq e
    W$. Every box during this round, must have been heavy (at least
    for a little while), which implies that it had weight
    $\geq \eps W$, as weights only increase during the algorithm
    execution. This implies that all the boxes are
    $\tfrac{\eps}{e}$-heavy at the end of the round. Thus, the weak
    $\tfrac{\eps}{e}$-net the algorithm computes must pierce all the
    boxes of $\Boxes$.
\end{proof}

The above lemma shows that the algorithm succeeds if it stops, and it
must stop in a stage if $k$ is sufficiently large.

\begin{theorem}%
    \thmlab{main:1}
    Let $\Boxes$ be set of $n$ axis-aligned boxes in $\Re^d$, for some
    fixed $d\geq 2$, and let $\popt = \PNX{\Boxes}$ be the piercing
    number of $\Boxes$. The above algorithm computes, in
    $O(n^{(d+1)/2}\log ^3 n)$ expected time, a piercing set of
    $\Boxes$ of size $O(\popt\log\log \popt)$.
\end{theorem}
\begin{proof}
    The above implies that if the algorithm outputs a piercing set,
    then it is of the desired size (it might be that the algorithm
    stops at an earlier stage than expected with a guess of $k$ that is
    smaller than $\popt$). Thus, it must be that if the guess for the
    value of $k$ is too small, then the algorithm double weights in
    vain, and performs too many rounds (i.e., their number exceeds
    $\tau$), and the algorithm continues to the next stage.

    Let $h = \ceil{ \log \popt}$ and $\mm=|\VV|$. Overall, the
    algorithm performs at most
    \begin{equation}
        \eqlab{doubling}%
        \sum_{i=1}^h O( 2^i \log \mm) = O( \popt \log \mm)
        =%
        O( \popt d \log n)
    \end{equation}
    doubling operations (i.e., \double). Every round requires $n$
    \weight{} operations. There are $O( \log \mm )$ rounds in each
    stage, and there are $h$ stages. We conclude that the algorithm
    performs
    \begin{equation*}
        O(n \log \mm \log \popt) = O(n \log^2 n )
    \end{equation*}
    \weight operations. Thus, the overall running time (including preprocessing and activating all boxes), is
    \begin{equation*}
        O(n^{(d+1)/2}\log n+n\cdot n^{(d-1)/2}\log n\cdot \log^2 n)
        =
        O(n^{(d+1)/2}\log^3 n).
    \end{equation*}
    Finally, the algorithm in \secref{weak-net} for computing a weak
    $\eps$-net first chooses a random sample $\Q$ of $\VV$ of size
    $O(\eps^{-1}\log \eps^{-1})$ and then computes a weak
    $\frac{\eps}{2}$-net of $(\Q,\Boxes)$. Using $\sample$, $\Q$ can
    be computed in $O(|\Q|\log n)$ time. Combining this with
    \lemref{net}, we conclude that a $\eps/e$-net of $(\VV,\Boxes)$ of
    size $O(\popt\log\log\popt)$ can be computed in
    $O\pth{(n+\eps^{-1})\log^d n}$ time, which is dominated by the
    data-structure's overall running time.
\end{proof}

\subsection{An Improved \MWU Algorithm}
\seclab{improved} We now present an alternative algorithm for piercing
that exploits \LP duality. Operationally, the improved algorithm
differs from the basic one in a few ways. The algorithm does not use
the implicit arrangement data-structure to maintain the weights on
$\VV$ directly. Moreover, instead of doubling the weight of light
boxes one at a time, the algorithm performs ``batch doubling'', i.e.,
it doubles the weights of a collection of boxes at the same time. The
algorithm also does not compute the weights of boxes exactly. Instead,
it approximates the weights using a suitable sample. This sample is
periodically recomputed by a process we call ``batch sampling'' that
uses the implicit arrangement data structure in $\Re^{(d-1)}$. Before
describing the improved algorithm, we explore the duality between
piercing sets and independent sets, and discuss the details of batch
sampling.

\begin{figure}
\begin{tabular}{cc}
  \fbox{%
  \begin{minipage}{0.4\linewidth}
      \begin{equation*}
          \begin{aligned}
            \fPN(\Boxes,\S)=
            &\min
              \sum_{v \in \S} x_v
            \\
            \forall \bb \in \Boxes\qquad
            & \sum_{v \in \bb \cap \S} x_v \geq 1
            \\
            \forall v \in \S  \qquad & x_v \geq 0.
          \end{aligned}
      \end{equation*}
  \end{minipage}}%
  & %
    \fbox{%
    \begin{minipage}{0.4\linewidth}
        \begin{equation*}
            \begin{aligned}
              \fIS(\Boxes,\S)=
              &\max
                \sum_{\bb \in \Boxes} y_\bb^{}
              \\
              \forall v \in \S\qquad
              & \sum_{\bb \in \Boxes: v \in \bb} y_\bb^{} \leq 1
              \\
              \forall \bb \in \Boxes  \qquad & y_\bb^{} \geq 0.
            \end{aligned}%
        \end{equation*}
    \end{minipage}%
    }%
  \\[1.2cm]
    \begin{minipage}{0.45\linewidth}
        \smallskip%
        The \LP for the fractional piercing number.
    \end{minipage}
  &
    \begin{minipage}{0.45\linewidth}
        \smallskip%
        The dual \LP for the fractional independence number.
    \end{minipage}
\end{tabular}
\caption{\LP{}s for the stabbing and independence problems. By
   duality, we have
   $\fPN(\Boxes,\S)=\fIS(\Boxes,\S)$.}
\figlab{lp:s}
\end{figure}

\paragraph{\LP duality, piercing and independence numbers}
\seclab{duality}

Let $\Boxes$ be a set of boxes in $\Re^d$, and let $\S$ be a set of
points in $\Re^d$. For the piercing problem, $\S=\VV$, the set of
vertices of the arrangement $\ArrX{\Boxes}$, but we describe the \LP
formulation in a more general setting because we need it later.  A
subset $\X\subseteq\Boxes$ is an \emphi{independent set} if no pair of
boxes of $\X$ contains the same point of $\S$. (Note that if $\S=\VV$
then the boxes in $\X$ are pairwise disjoint.) The \emphi{independence
   number} of $\Boxes$ with respect to $\S$, denoted by
$\ISX{\Boxes,\S}$, is the size of the largest independent set
$\X \subseteq \Boxes$.  Observe that
$\ISX{\Boxes,\S} \leq \PNX{\Boxes,\S}$, as each box in an independent
set must be stabbed, but no point can stab more than one box in the
independent set.

There are {\LP}s associated with stabbing and independence problems,
and they are dual of each other. See \figref{lp:s}. Let
$\fPN=\fPN(\Boxes,\S)$ denote the fractional piercing number of
$(\Boxes,\S)$, i.e., this is the minimum stabbing number of $\Boxes$
when solving the associated \LP, see \figref{lp:s}.  Similarly,
$\fIS=\fIS(\Boxes,\S)$ is the \emphi{fractional Independence number}
of $\Boxes$. By the strong \LP duality
$\fPN(\Boxes,\S) = \fIS(\Boxes,\S)$.

With this general formulation at hand, we return to the piercing
problem, i.e. $\S=\VV$. Let $\fPN=\fPN(\Boxes,\VV)$ and
$\fIS=\fIS(\Boxes,\VV)$. It is known, and also implied by
\lemref{eps:net:2d:log:log}, that
$\fPN \leq \popt \leq c_1 \fPN \log \log \fPN$, for some constant
$c_1 > 0$. This implies that $\fPN \geq \popt /(c_1 \log \log
\popt)$. We need the following standard algorithms for computing
independent set and piercing set of boxes. Better results are known
\cite{cc-misr-09,gkmmp-taami-22, m-amisr-21}, but they are
unnecessary for our purposes.
\begin{lemma}[\tcite{aks-lpmis-98}]
    \lemlab{indep}%
    Let $\Boxes$ be a set of $n$ boxes in $\Re^d$. An independent
    subset of boxes of $\Boxes$ of size $\Omega(\fIS /\log^{d-1}n)$
    can be computed in $O(n \log^{d-1} n )$ time.
\end{lemma}

\begin{proof}
    This algorithm is well known and we sketch it here for
    completeness -- it works by induction on the dimension. For $d=1$, a
    greedy algorithm picking the first interval with the minima right
    endpoint computes the optimal independent set --- the running time
    of this algorithm is $O(n)$ after sorting. It is straightforward
    to verify that one can assume that the \LP solution in this case
    is integral, and thus $\fIS = \ISX{\Boxes}$.

    For $d=2$, compute the vertical line such that its $x$-coordinate
    is the median of the $x$-coordinates of the endpoints of the
    boxes. Compute the optimal independent set of rectangles
    intersecting this line (which is just the one dimensional
    problem), and now recurse on the two subsets of rectangles that do
    not intersect the median line. This results in a partition of the
    set of rectangles $\Boxes$ into $O(\log n)$ sets, such that for
    each set, we have the optimal (fractional) independent
    set. Clearly, one of them has to be of size $\Omega(\fIS/\log n)$.

    For $d>2$, the same algorithm works by approximating the optimal
    solution along the median of the first coordinate (i.e., $d-1$
    subproblem), and then recursively on the two subproblems. The
    bounds stated readily follows.
\end{proof}

The piercing problem can be solved exactly by the greedy algorithm in
one dimension (i.e., add the leftmost right endpoint of an input
interval to the piercing set, remove the intervals that intersect it,
and repeat). For higher dimensions, one can perform the same standard
divide and conquer approach, used in \lemref{indep}, and get the
following.
\begin{lemma}[\tcite{n-fsbhd-00}]
    \lemlab{pierce}%
     Let $\Boxes$ be a set of $n$ boxes in $\Re^d$
    for $d\geq2$. A piercing set $\P$ of $\Boxes$ of size
    $O(\fPN \log^{d-1}n)$ can be computed in $O(n \log^{d-1} n )$
    time.
\end{lemma}

\begin{corollary}
    \corlab{indep:good}%
    Let $\Boxes$ be a set of $n$ boxes in $\Re^d$ for $d\geq2$. Then
    the following two sets can be computed in $O( n \log^{d-1} n)$
    time:
    \begin{compactenumi}
        \smallskip%
        \item An independent set $\Boxes' \subseteq \Boxes$ of
        $\Boxes$ of size
        $ \Omega\bigl( \tfrac{\popt}{ \log^{d-1} n \log \log
           n}\bigr)$, and

        \smallskip%
        \item a piercing set $\P \subseteq \Re^d$ of $\Boxes$ of size
        $O( \popt \log^{d-1} n)$.
    \end{compactenumi}
\end{corollary}

\begin{proof}
    By \lemref{indep} one can compute, in $O(n\log^{d-1} n)$ time, an
    independent set $\Boxes'$ of size
    $\Omega\bigl( \frac{\fIS}{\log^{d-1}n} \bigr)$.  By
    \lemref{pierce}, a piercing set $\P$ of size
    $O(\popt\log^{d-1} n)$ can be computed in $O(n\log^{d-1}n)$
    time. By \lemref{eps:net:2d:log:log}, we have
    \begin{equation*}
        \cardin{\Boxes'}
        =%
        \Omega\pth{ \frac{\fIS}{\log^{d-1}n}}
        =
        \Omega\pth{ \frac{\popt}{\log^{d-1}n \log \log n}},
    \end{equation*}
\end{proof}

\paragraph{Batch sampling}

For the purposes of our improved algorithm, we need to support the
following ``batch sampling'' operation.  Given a set $\Boxes$ of $n$
boxes in $\Re^d$, a multiset (i.e., list)
$\ListB \subseteq \Boxes$ of boxes such that $\cardin{\S}=O(n\log n)$, and a parameter $r>0$ such that $r=O(n\log n)$, the task at hand is to compute a random subset $\RSample$ of
$r$ vertices of $\ArrX{\Boxes}$, where each vertex $v\in\ArrX{\Boxes}$ is sampled independently with probability $\hwY{p}{\S} /\hwY{\VV}{\S}$, see \Eqref{doubling:w}.
This procedure can be implemented using the data structure
described in \defref{arr:ds1}, but one can do better, by sweeping along the $x_d$-axis and constructing the
data structure in one lower dimension, as follows.

Let $\hBoxes$ be the collection of $(d-1)$-dimensional projections
of the boxes in $\Boxes$, to the hyperplane $x_d=0$. Let $\VV'$
be the set of vertices in $\ArrX{\smash{\hBoxes}}$, the
arrangement of $\hBoxes$ in $\Re^{d-1}$.  Let $\Family $ be the
set of $x_d$-coordinates of all the vertices of the boxes of
$\Boxes$.

Build $\DS$, an instance of the $(d-1)$-dimensional data structure of
\defref{arr:ds} on $\hBoxes$.  Importantly, initially all the boxes of
$\hBoxes$ are inactive.  Observe that
$\VV \subseteq \VV'\times\Family$. Assume $\Family$ is sorted and the
ith point (in sorted order) is denoted by $e_i$. We perform two
space-sweeps along the $x_d$-axis. In step $i$ of a sweep, $\DS$
represents the point-set $\J_i=\pth{\VV'\times\{\e_i\}} \cap \VV$
lying on a $\Re^{(d-1)}$-dimensional subspace, and we compute
$\alpha(i)$, the total weight of vertices in $\J_i$

At any step $i$ of the first sweep, for every box $\bb \in \Boxes$
that starts (resp. ends) at $e_i$ call $\Insert(\bb)$ (resp. $\Delete(\bb)$). This activates/deactivates all the vertices on the
boundary on $\bb$. Similarly, for every $\bb\in \ListB$ that
begins (resp. ends%
) at $e_i$, use $\double(\bb)$ (resp. $\halve(\bb)$) operation on
$\DS$ to double (resp. halve) the weight of the points in
$\J_i\cap\bb$. Next, use the $\weight$ operation on $\DS$ to get
the weight of all the points in $\J_i$ and store it in
$\alpha[i]$.
At the end of the first sweep, independently sample $r$
numbers from the set $\{1,...,2n\}$ where the integer $j$ is
sampled with probability
${\alpha(j)}/{\sum\nolimits_{i=1}^{2n} \alpha(i)}$. Let $\delta(j)$
denote the number of times $j$ was sampled.

Next, reset $\DS$ and begin the second sweep. In step $i$ of the
second sweep, perform the same operations except replace the
weight-computing step with taking $\delta(i)$ independent samples from
$\J_i$ using the $\sample$ operation on $\DS$.

Observe that ${\alpha(i)}/{\sum\nolimits_{j=1}^{2n} \alpha(j)}$ is the
total probability mass of all the points in $\J_i$ at the end of the
first sweep. Using this, it is easy to verify that at the end of the
second sweep the sampled points correspond to the desired subset
$\RSample$.

Initializing $\DS$ takes $O(n^{{d}/{2}}\log n)$ time, and performing
each step of the sweep takes
\begin{equation*}
        O(n^{(d-2)/2}\log n)
\end{equation*}
time. Each sweep handles $O(|\ListB|)=O(n\log n)$ events, hence,
$\RSample$ can be computed in $O(n^{d/2}\log^2 n)$ time. We thus
obtain the following:

\begin{lemma}
    \lemlab{improved:rel_sample}
    Let $\Boxes$ be a set of $n$ boxes in $\Re^d$, and let $\ListB$ be
    a multiset of boxes such that $\cardin{\ListB} = O( n \log
    n)$. Let $\wC$ be the doubling weight function induced by $\ListB$
    over the vertices $\VV$ of $\ArrX{\Boxes}$, see
    \Eqref{doubling:w}. Given a parameter $r>0$ such that $r=O(n\log n)$, a random subset
    $\RSample\subset \VV$ of size $r$, where each vertex of $\VV$ is sampled
    with probability proportional to its weight, can be computed in  $O(n^{{d}/{2}}\log^{2} n)$ time.
\end{lemma}

\paragraph{The new \MWU algorithm}
\seclab{fewer:iterations}

As in the previous algorithm, the algorithm performs an exponential
search on $k$ until $k\leq 2\popt$. For a particular guess $k$, the
new algorithm also works in rounds. The difference is how each round
is implemented. Instead of doubling the weight of a light box as soon
as we find one, we proceed as follows. Let
$\wC:\VV \rightarrow\Re_{\geq 0}$ be the weight function at the
beginning of the current round. Observe that for any point $p\in\VV$, $\wC(p)=\hwY{p}{\ListB}$ where $\ListB$ is the multi-set of boxes that doubled their weights so far, see \Eqref{doubling:w}.
At the beginning of each round, we process
$\Boxes$ and $\ListB$ to generate a random subset $\RSample \subset\VV$ of
size $O(k\log n)$, such that for a box $\bb\in \Boxes$ one can
determine whether it is light, i.e., $\wC(\bb)\leq \wC(\VV)/4k$ by
checking if $|\bb\cap \RSample|\leq |\RSample|/4k$.  By processing
$\RSample$ into an orthogonal range-counting data structure
\cite{a-rs-04}, we can compute $|\bb\cap\RSample|$, in
$O(\log^{d-1} n)$ time, for any $\bb\in\Boxes$, by querying the data
structure with $\bb$.  By repeating this for all boxes of $\Boxes$, we compute a set $\LightB \subseteq \Boxes$ of light boxes.

Next, we compute an independent set of boxes
$\IndepB \subseteq \LightB $ using \lemref{indep}. There are several
possibilities:
\begin{compactitem}
    \smallskip
    \item If $\cardin{\IndepB} \geq 2k$, then the guess for $k$ is too
    small. We double the value of $k$, and restart the process.

    \smallskip
    \item If $\cardin{\IndepB} < c_1k/ \log^{2d-1} n$, then by
    \corref{indep:good}, a piercing set $\P$ of $\LightB$ of $O(k)$
    points can be computed in $O(n\log^{d-1} n)$ time. The remaining
    boxes of $\HeavyB = \Boxes \setminus \LightB$ are (say)
    $1/4.01k$-heavy.  By applying \lemref{net} to $\HeavyB$ and
    $\RSample$, we compute a weak $1/4.01k$-net $\N$ for these boxes
    of size $O( k \log \log k)$ in $O( n \log^d n)$ expected time. The
    set $\P\cup \N$ is the desired piercing set of $\Boxes$ of size
    $O(\popt\log\log\popt)$.

    \smallskip
    \item If $\cardin{\IndepB} \geq c_1k/ \log^{2d-1} n$, then we
    update $\ListB$, the multi-set of boxes doubled so far, to
    $\ListB = \ListB \cup \IndepB$. The algorithm now continues to the
    next round.
\end{compactitem}

\begin{remark}
    Observe that the batch doubling operation in the above algorithm
    corresponds to merely updating the list $\ListB$. The actual work
    associated with the doubling is in regenerating the sample
    $\RSample$ at the beginning of the next round, which is done using
    batch sampling.
\end{remark}

\paragraph{Computing the random sample $\RSample$ in each round}
\seclab{relative:approx}

Let $\wC:\VV \rightarrow\Re_{\geq 0}$ be the weight function in the
beginning of a particular round. Our goal is to compute a sample
$\RSample \subset \VV$, such that we can check if a box $\bb\in\Boxes$
is light by checking if $|\bb\cap \RSample|\leq
{|\RSample|}/{4k}$. Recall that $\ListB$ denotes the list of boxes for
which the \MWU algorithm had doubled the weights. Moreover, recall
that the weight of a vertex $v \in \VV$ is $2^{\nrstY{\ListB}{v}}$,
where $\nrstY{\ListB}{v}$ denotes the number of boxes in $\ListB$ that
contain $v$. We use the concept of relative approximations.

\begin{defn}
    Let $\Sigma=(\X,\RangeSet)$ be a finite range space, and
    $\RSample \subset \X$, and let
    $\wC: \X \xrightarrow{}\Re_{\geq 0}$ be a weight function. For a
    range $\range \in \RangeSet$, let
    \begin{math}
        \Measure{\range}%
        =%
        {\wC({\range \cap \FGroundSet})}/{\wC({\FGroundSet})}
    \end{math}
    and
    \begin{math}
        \sMeasure{\range}%
        =%
        {\cardin{\range \cap \RSample}/}{\cardin{\RSample}}.
    \end{math}
    Then, $\RSample$ is a \emphi{relative $(p,\eps)$-approximation} of
    $\Sigma$ if for each $\range\in \RangeSet$ we have:
    \begin{compactenumi}
        \smallskip%
        \item If $\Measure{\range} \geq p$, then
        $\displaystyle \displaystyle (1-\eps) \Measure{\range} \le
        \sMeasure{\range}\leq (1+\eps) \Measure{\range}$.

        \smallskip%
        \item If $\Measure{\range} \leq p$, then
        $\sMeasure{\range} \leq (1+\eps)p$.
    \end{compactenumi}
\end{defn}

It is known \cite{hs-rag-11, h-gaa-11} that if the VC-dimension of
$\Sigma$ is $\Dim$, then a random sample $\RSample$ of size
$ O \bigl( \tfrac{1}{\eps^2p} \bigl[\Dim \log\frac{1 }{ p} +
\log\frac{1}{\BadProb} \bigr] \bigr)$, where each point is chosen with
probability proportional to its weight, is a relative
$(p,\eps)$-\si{approx}\-\si{imation} with probability
$\geq 1-\BadProb$. In view of this result, we can detect light boxes
of $\Boxes$ (with respect to $\S$) as follows. Set $\eps=0.01$,
$p=\frac{\eps}{20k}$, $\BadProb=\frac{1}{n^{O(1)}}$. We chose a random
subset $\RSample \subseteq \VV$ of $r=O(\frac{k}{\eps^2}\log n)$
points. Then we have the following property for each box
$\bb\in\Boxes$.

\begin{compactenumI}
    \smallskip
    \item \itemlab{c:i} %
    If
    $\Measure{\bb} = \wC({\bb \cap \VV}) /\wC(\Boxes) \geq (\eps/20k)$
    then
    $0.99 \Measure{\bb} \leq \sMeasure{\bb} \leq 1.01 \Measure{\bb}$.

    \smallskip
    \item \itemlab{c:ii} %
    If $\Measure{\bb} \leq \eps/20k$ then
    $ \sMeasure{\bb} \leq 1.01/20k$.
\end{compactenumI}
\smallskip%

Therefore to check if a box is $\bb$ is light, it suffices to check
$|\bb\cap\RSample|$. As for computing the sample $\RSample$, observe
that this is exactly what batch sampling is designed for, see
\lemref{improved:rel_sample}.

\paragraph{Analysis}

The correctness of the new \MWU follows from the previous one. We now
analyze the running time. Each round except the last round doubles the
weight of $\Omega(k/ \log^{2d} n)$ boxes. Since the total number of
weight-doubling operations performed by the algorithm is
$O(\popt \log n)$, see \cite{ap-naghs-20}, the algorithm stops within
$O( \log^{2d} n)$ rounds. In a particular round, the algorithm uses
batch sampling to recompute the random subset $\RSample$ which it uses
to identify the set $\LightB\subseteq \Boxes$ of light
boxes. \lemref{improved:rel_sample} shows that the cost of batch
sampling is bounded by $O(n^{{d}/{2}}\log^{2} n)$ time. As discussed
above, once $\RSample$ is computed, $\LightB$ can be identified in
$O(n\log^{d-1} n)$ time. Finding an independent set
$\IndepB\subseteq\LightB$ using \lemref{indep} also takes
$O(n\log^{d-1} n)$ time. Updating the multi-set $\ListB$ of boxes
whose weights have been doubled so far also takes $O(n)$ time. The
algorithm computes a piercing set and a weak net only in the last
round. Together, they take $O(n\log^{d} n)$ time.  The running time
for these steps is dominated by the time for the batch sampling in
each round.  Considering the need to do an exponential search for
$\popt$, we obtain the following:

\begin{theorem}
    \thmlab{pierce:alg}
    Let $\Boxes$ be set of $n$ axis-aligned boxes in $\Re^d$, for
    $d\geq 2$, and let $\popt=\PNX{\Boxes}$ be the piercing number of
    $\Boxes$. A piercing set of $\Boxes$ of size
    $O(\popt\log\log \popt)$ can be computed in
    $O(n^{d/2}\log ^{2d+3} n)$ expected time.
\end{theorem}

\section{Improving the running time}
\seclab{candidate}

Let $\Boxes=\{\bb_1,...,\bb_n\}$ be a set of axis-parallel $d$-boxes
in $\Re^d$, where a box
$\bb_i=\prod\nolimits_{t=1}^d [\alpha_t,\beta_t]$,
$\alpha_t,\beta_t\in \Re$, and $\alpha_t < \beta_t$. In this section,
we show that a subset $\P$ of $n^{O(1+\log d)}$ vertices of $\ArrX{B}$
can be computed in $n^{O(1+\log d)}$ time such that $(\P,\Boxes)$
contains a hitting set of size $(4d+1)2^{4d+1}\cdot \popt$, where
$\popt$ is the piercing number of $\Boxes$.  We can then use the
algorithm of Agarwal and Pan to compute a hitting set of
$(\P,\Boxes)$.

For $0\leq k\leq d$, a $(d-k)$-dimensional face $f$, or a
\emphi{$(d-k)$-face} for brevity, of $\bb_i$ is specified by fixing
$d-k$ coordinates to one of the endpoints of the corresponding
intervals. The \emphw{codimension} of $f$, denoted by $\codim( f)$ is
$k$.

For a $d$-dimensional box $\bb$, let $\CFaceX{\bb}$ be the collection
of its $\floor{d/2}$-faces. The codimension of these faces is
$\ceil{ d/2}$. Let
\begin{equation*}
    \psi
    =
    2^{\ceil{ d/2}} \binom{d}{\floor{ d/2 }}
    \leq%
    2^{3d/2}
\end{equation*}
be the number of $\floor{ d/2}$ faces in a $d$-box. For a set of
$d$-dimensional boxes $\Boxes'$, define
$\CFaceX{\Boxes'}=\bigcup\nolimits_{\bb\in \Boxes'}\CFaceX{\bb}$.

\begin{lemma}
    \lemlab{pairs}%
    For all $\bb,\bb' \in \Boxes$, with
    $\bb \cap \bb' \neq \emptyset$, there is a $\floor{ d/2 }$ face of
    $\bb$ or $\bb'$, such that $f \cap \bb \cap \bb' \neq \emptyset$.
\end{lemma}

\begin{proof}

    Let $v$ be a vertex of $\bb\cap\bb'$ and let $g$ and $g'$ be the
    lowest dimensional faces of $\bb$ and $\bb'$ respectively such
    that $v=g\cap g'$. Since $\codim(v)=d$,
    $\max\{\codim(g),\codim(g')\}\geq \ceil{ d/2 }$. Without loss of
    generality, assume that $\codim(g)\geq \ceil{ d/2}$, them
    $\dim(g)\leq \floor{ d/2 }$. The $\floor{d/2}$-dimensional face
    $f$ containing $g$ is the desired face with
    $f\cap\bb\cap\bb'\neq \emptyset$.
\end{proof}

\begin{corollary}
    \corlab{half:dim}%
    Let $\BSB$ be a set of $m \geq 3$ boxes in $\Re^d$ such
    $\bigcap \BSB \neq \emptyset$. Then there exists a box
    $\bb_j\in\BSB$, such that one of its $\floor{ {d}/{2}}$-faces $f$
    intersects at least ${m}/(3\psi)$ boxes in $\BSB$.
\end{corollary}
\begin{proof}
    Let $\F=\CFaceX{\BSB}$ be the set of $\floor{d/2}$-faces of boxes
    in $\BSB$. We have $|\F|\leq m\psi$. By \lemref{pairs}, for any
    pair $\bb,\bb'\in\BSB$, there is a face $f\in\F$ such that
    $f\cap\bb\cap\bb'\neq \emptyset$. Hence, by the pigeonhole
    principle, there is a face in $\F$ that intersects at least
    \begin{equation*}
        \frac{\binom{m}{2}}{|\F|}%
        =%
        \frac{m(m-1) }{2\cardin{\F}}
        \geq
        \frac{m(m-1)}{2m \psi}
        \geq
        \frac{m}{3\psi}
    \end{equation*}
    boxes of $\BSB$.
\end{proof}

For a subset $\BSB\subseteq \Boxes$, a vertex $v\in\ArrX{\BSB}$ is
\emphw{defined} by $k$ faces of $\CFaceX{\BSB}$ if $v$ is a vertex in
the arrangement of a subset of $k$ faces of $\CFaceX{\BSB}$.

\begin{lemma}
    \lemlab{sparse}%
    Let $\BSB=\{\bb_1,...,\bb_m\}$ be a set of $m$ boxes in $\Re^d$
    such $\bigcap\BSB \neq \emptyset$. Then there exists a vertex of
    $\ArrX{\BSB}$ that is contained in at least $m/2^{4d}$ boxes of
    $\BSB$ and is defined by a subset of $ 1 + \floor{\log_2 d}$ faces
    of $\CFaceX{\BSB}$.
\end{lemma}
\begin{proof}
    The proof is by induction on $d$. For $d=2$, there is clearly a
    vertex formed by the intersection point of two edges in $\BSB$
    that lies on the boundary of $\bigcap\nolimits_{i=1}^m\bb_i$. Thus
    the claim holds.

    For $d >2$, by \corref{half:dim}, there exists a $\floor{d/2}$
    face of a box in $ \BSB$ that intersects at least $m/ \psi$ boxes
    of $\BSB$. Let $\BSC\subseteq\BSB$ be the subset of these
    boxes. Observe that the boxes of $\BSC$ have a common intersection
    point in $f$, i.e.
    $(\bigcap\nolimits_{\bb\in\BSC} \bb)\cap f\neq \emptyset$.  Let
    $h$ be the $\floor{d/2}$-flat that supports $f$. We apply the
    claim recursively to the set of $\floor{d/2}$-boxes
    $\BSC \cap h = \Set{ \bb \cap h}{\bb \in \BSC}$ on $h$. Clearly,
    this recursion terminates with a vertex $v$.

    Let $\tau(d)$ be the number of boxes of $\BSB$ defining $v$. Then
    we have the recurrence

\[
\tau(d) =
\begin{cases}
  1, & d = 1, \\
  2, & d = 2, \\
  \tau(\floor{ d/2 })+1, & d > 2.
\end{cases}
\]
    The solution of the recurrence is clearly $1+\floor{\log_2 d}$.
    Let $u(d)$ be the fraction of boxes of $\BSB$ that contain
    $v$. We now have the following recurrence,
    \begin{equation*}
        u(d) =
        \begin{cases}
          1, &  d = 1,2 \\
          \displaystyle \frac{u(\floor{ d/2 })}{\psi(d)}, & d > 2.
        \end{cases}
    \end{equation*}
    Since $\psi\leq 2^{3d/2}$, it is easily seen that
    \begin{math}
        u(d) \geq
        \frac{1}{3^{\log_2 d}\cdot 2^{3d}} \geq 1/2^{4d}.
    \end{math}
\end{proof}

For a set $\BS$ of boxes in $\Re^d$, let
\begin{equation}
    \VSetX{\BS} = \bigcup_{\BSB \subseteq \BS\,:\, \cardin{\BSB}
       \leq 1 + \floor{\log_2 d} } \VX{\ArrX{\BSB}}.
    \eqlab{V:Set}
\end{equation}
Observe that
$\cardin{\VSetX{\BS}} = O( \cardin{\BS}^{1 + \floor{\log_2 d}})$.

\begin{lemma}
    \lemlab{replacement}%
    Let $\BSB$ be a set of boxes in $\Re^d$ such that $\bigcap \BSB\neq \emptyset$.  Then, there exists a set $\S \subseteq \VSetX{\BSB}$
    that stabs all the boxes of $\BSB$, and $|\S|=(4d+1)\cdot 2^{4d+1}$.
\end{lemma}
\begin{proof}
    Let $\P = \VSetX{\BSB}$. Clearly, $\P$ stabs all the boxes in
    $\BSB$. Consider the piercing \LP and its dual \LP described in
    \secref{improved} (see \figref{lp:s}) \emph{piercing linear
       program}. Let $\alpha$ be the fractional independence number,
    i.e., $\alpha=\fIS (\BSB)$, the optimal solution of the dual \LP.

    Let $\omega : \BSB \to \Re$ represent the weight assigned to each
    box in the optimal solution of the dual \LP. Take a random sample
    $\BSC \subseteq \BSB$, of size $N = c \fIS \log \fIS$, where each
    sample is independent and each box $\bb \in \BSB$ is sampled
    proportional to its weight. Specifically, for a box $\bb \in \BSB$
    assigned weight $\omega(\bb)$, it is sampled with probability
    $\omega(\bb) / \fIS$. Note that
    $\sum\nolimits_{\bb \in \BSB} \omega(\bb) = \fIS$.

    For a point $p$, let $\BSB \sqcap p = \Set{\bb\in\BSB}{p\in \bb}$,
    and consider the dual set system
    \begin{equation*}
        \Sigma^*=(\BSB,\Set{ \BSB \sqcap p}{ p\in \P }),
    \end{equation*}
    where each range is the subset of boxes pierced by a point of
    $\P$. By lifting each box $\bb$ in $\BSB$ to a point
    $\bb^*\in\Re^{2d}$ and each point in $p\in\P$ to a quadrant $p^*$
    in $\Re^{2d}$ such that $p\in \bb$ if and only if $\bb^*\in p^*$,
    we can show that the VC-dimension of $\Sigma^*$ is bounded by
    $4d$. Thus, by selecting the constant $c$ appropriately while
    fixing $N$ (where $c$ is a small absolute constant depending only
    on $d$), it can be ensured that $\BSC$ forms a
    $(p, \eps)$-relative approximation \cite{h-gaa-11,hs-rag-11}, with
    $p = 1 / (2\alpha)$ and $\eps = 1 / 4$.

    Since for any $p \in \P$, we have
    $\frac{1}{\alpha}\cdot\sum\nolimits_{\bb \in \BSB: p \in \bb}
    {\omega(\bb)} \leq \frac{1}{\alpha}$, it follows that for any
    $p \in \P$, the set satisfies
    \begin{math}
        \frac{\cardin{ \BSC \sqcap p } }{N} \leq%
        \frac{1 + {1}/{4}}{\alpha}.
    \end{math}
    This implies that for any $p \in \P$,
    \[
        \cardin{\BSC \sqcap p}%
        \leq%
        \Bigl(1 + \frac{1}{4}\Bigr)c \log \alpha \leq 2c \log \alpha.
    \]

    On the other hand, by \lemref{sparse}, there is a point $p \in \P$
    that stabs at least $N/2^{4d}$ boxes of $\BSC$. We thus have the
    inequality $N/2^{4d} \leq 2c \log \alpha$, which implies
    \begin{equation*}
        \frac{c}{2^{4d}}\alpha \log \alpha
        \leq
        2c \log \alpha
        \qquad \text{or}\qquad
        \alpha \leq 2^{4d+1}.
    \end{equation*}

    As mentioned above, by strong duality the fractional stabbing number of $(\P, \BSB)$ is at most
    $ \alpha \leq 2^{4d+1}$. However, standard $\eps$-net argument readily
    implies that a sample of size $O( \alpha \log \alpha) = O(1)$ of points from
    $\P$ (sampled according to these fractional weights) is a stabbing
    set for $\BSB$ with constant probability. Let $\S$ be such a stabbing set of size $O\bigl((4d+1)\cdot 2^{4d+1}\bigr)$.  All the boxes of
    $\BSB$ are stabbed by points of $\S$.
\end{proof}

\begin{corollary}
    Let $\BS=\{\bb_1,...,\bb_n\}$ be a set of boxes in $\Re^d$. Let
    $\popt$ be the size of the smallest piercing set of $\BS$.  A set
    $\Ps = \VSetX{\BS} \subset \Re^d$ can be computed, in
    $O(n^{\floor{\log_2 d}+1})$ time, such that the range space
    $(\Ps,\BS)$ has a hitting set of size $O\bigl((4d+1)\cdot 2^{4d+1}\cdot \popt \bigr)$.
\end{corollary}
\begin{proof}
    Let $\P^*$ be a piercing set of $\BS$ of size $\popt$. By
    \lemref{replacement}, for each point $ p \in \P^*$, there is a set of $O\bigl((4d+1)\cdot 2^{4d+1}\cdot \popt \bigr)$ points in $\Ps$ that stabs all the boxes of $\BS$ that $p$
    stabs (it potentially stabs more boxes). This readily implies the
    bound on the size of the hitting set. As for the running time, it
    follows readily from computing the set of vertices of the
    arrangement $\ArrX{\BSB}$, over all subsets $\BSB \subseteq \BS$
    of size $1 +\floor{\log_2d} $, see \Eqref{V:Set}.
\end{proof}

By combining the above corollary with Agarwal and Pan's \cite{ap-naghs-20} algorithm for computing hitting sets for boxes we obtain the following theorem:

\begin{theorem} \thmlab{thm:fast-pierce}

Let $\BS=\{\bb_1,...,\bb_n\}$ be a set of boxes in $\Re^d$. Let
    $\popt$ be the size of the smallest piercing set of $\BS$. A piercing set of $\BS$ of size $O(d^2\cdot 2^{4d}\cdot\popt\log\log\popt)$ can be computed in $O(n^{\floor{\log_2 d}+1})$ expected time.

\end{theorem}

\section{Multi-round piercing algorithm}
\seclab{iteralg}

Let $\Boxes$ be a set of (closed) boxes in $\Re^d$, and let
$\VV = \VX{\Boxes}$ be the set of vertices of the arrangement
$\ArrX{\Boxes}$. When considering a piercing set for $\Boxes$, one can
restrict the selection of piercing points to points of $\VV$.

The key insight is that a piercing set for a sufficiently large, but
not too large, sample is a piercing set for almost all the boxes.

\begin{lemma}
    \lemlab{sampling}%
    Let $\Boxes$ be a set of $n$ boxes in $\Re^d$, $\delta \in (0,1)$
    and $t > 0$ be parameters, and let $\Sample \subseteq \Boxes$ be a
    random sample of size $O( \tfrac{dt}{\delta} \log n)$. Assume
    there is a set $\Q$ of $t$ points that stabs all the boxes of
    $\Sample$. Then, at most $\delta n$ boxes of $\Boxes$ are not
    pierced by $\Q$, and this holds with probability
    $\geq 1 - 1/n^{O(d)}$.
\end{lemma}

\begin{proof}
    For a set $P$ of points in $\Re^d$, let
    $\Boxes \setminus P = \Set{ \bb \in \Boxes}{ \bb \cap P =
       \emptyset}$ be the leftover set of all the boxes not stabbed by
    points of $P$. The family of all ``large'' leftover sets is
    \begin{equation*}
        \Family =%
        \Set{\Boxes \setminus P}{P \subseteq \VV,
           \cardin{P} \leq t
           \text{ and }
           \cardin{\Boxes \setminus P} \geq \delta n
        },
    \end{equation*}
    where $\VV = \VX{\Boxes}$.  As $\cardin{\VV} =O(n^d)$, it follows
    that $\cardin{\Family} = O(n^{dt})$.

    Fix a ``bad'' set $\BoxesB \in \Family$.  Let
    $u = c\tfrac{dt}{\delta} \ln n$ be the size of $|\Sample|$, where
    $c$ is a sufficiently large constant.  If
    $\Sample \cap \BoxesB \neq \emptyset$, then $\Boxes \setminus \Q$
    can not be $\BoxesB$, as $\Q$ starts all the boxes of $\Sample$.
    The probability that $\Sample$ misses all the boxes of $\BoxesB$
    is at most
    \begin{equation*}
        \psi
        =%
        \pth{ 1- \delta}^u
        \leq%
        \exp\pth{ -\delta u }
        =%
        \exp( - c\cdot d\cdot t \cdot \ln n)
        =%
        \frac{1}{n^{cdt}}.
    \end{equation*}
    In particular, by the union bound, the probability that $\Sample$
    misses any of the sets of $\Family$ is at most
    $\cardin{\Family'} \psi < 1/n^{O(d)}$, for $c$ sufficiently large.
    Namely, the residual set of boxes $\Boxes \setminus \Q$ is not in
    $\Family$, and it thus has size $< \delta n$,
\end{proof}

\lemref{sampling} suggests a natural algorithm for piercing -- pick a
random sample from $\Boxes$, compute (or approximate) a piercing set
for it, compute the boxes this piercing set misses, and repeat the
process for several rounds. In the last round, hopefully, the number
of remaining boxes is sufficiently small, that one can apply the
piercing approximation algorithm directly to it. We thus get the
following.

\begin{lemma}
    \lemlab{multi:round}%
    Let $\Boxes$ be a set of $n$ boxes in $\Re^d$, and let
    $\nRounds > 1$ be a parameter. Furthermore, assume that we are
    given an algorithm \AlgPierce that for $m$ boxes, can compute a
    $O( \log \log \popt )$ approximate to their piercing set in
    $\TPierce(m)$ time, where $\popt$ is the size of the optimal
    piercing set for the given set. Then, one can compute a piercing
    set for $\Boxes$ of size $O( \nRounds \popt \log \log \popt)$ in
    expected time
    \begin{equation*}
        O\pth {\nRounds \TPierce\bigl( \popt^{1-{1}/{\nRounds}}
           n^{1/\nRounds} \log n \bigr) + n \log^{d-1} \popt }.
    \end{equation*}
\end{lemma}
\begin{proof}
    We assume that we have a number $k$ such that
    $\popt \leq k \leq 2\popt$, where $\popt$ is the size of the
    optimal piercing set for $\Boxes$. To this end, one can perform an
    exponential search for this value, and it is easy to verify that
    this would not effect the running time of the algorithm.

    The algorithm performs $\nRounds$ rounds. Let $\Boxes_0 =
    \Boxes$. Let $\delta = {(k/n)}^{1/\nRounds}$. In the $i$\th round,
    for $i=1, \ldots, \nRounds-1$, we pick a random sample $\Sample_i$
    from $\Boxes_{i-1}$ of size
    \begin{equation*}
        m
        =%
        O\Bigl(  d k{\delta}^{-1} \log n \Bigr)
        =%
        O\bigl(  d k^{1-{1}/{\nRounds}} n^{1/\nRounds} \log n \bigr)
        =%
        O\bigl(  \popt^{1-1/\nRounds} n^{1/\nRounds} \log n \bigr).
    \end{equation*}
    In the $\nRounds$\th round, we set
    $\Sample_{\nRounds} = \Boxes_{\nRounds-1}$. Now, we approximate
    the optimal piercing set for $\Sample_i$, by calling
    $\Q_i \leftarrow \AlgPierce( \Sample_i)$. If the piercing set
    $\Q_i$ is too large -- that is, $\cardin{\Q_i} \gg k \log \log k$,
    then the guess for $k$ is too small, and the algorithm restarts
    with a larger guess for $k$. Otherwise, the algorithm builds a
    range tree for $\Q_i$, and streams the boxes of $\Boxes_{i-1}$
    through the range tree, to compute the set
    $\Boxes_i = \Boxes_{i-1} \setminus \Q_i$, the boxes in
    $\Boxes_{i-1}$ not pierced by $\Q_i$. By \lemref{sampling}, we
    have
    $\cardin{\Boxes_{i}} \leq \delta \cardin{\Boxes_{i-1}} \leq
    \delta^i n$ with high probability. If
    $\cardin{\Boxes_{i}}> \delta \cardin{\Boxes_{i-1}}$, we repeat
    round $i$, so assume
    $\cardin{\Boxes_i}\leq \delta\cardin{\Boxes_{i-1}}$. In particular
    $\cardin{\Sample_{\nRounds}}=\cardin{\Boxes_{\nRounds-1}}\leq
    \delta^{\nRounds-1}n\leq
    \popt^{1-1/\nRounds}n^{1/\nRounds}$. Therefore, the total time
    spent by $\AlgPierce(.)$ in $\nRounds$ rounds is
    $O(\nRounds\TPierce(\popt^{1-1/\nRounds}n^{1/\nRounds}\log
    n))$. Finally, during the first $\nRounds-1$ iterations, computing
    $\Boxes_i$ takes
    \begin{equation*}
        \sum_{i=1}^{\nRounds-1}
        O( \cardin{\Boxes_{i-1}} \log^{d-1} \cardin{\Q_i})
        =%
        \sum_{i=1}^{\nRounds-1}
        O( \delta^i n \log^{d-1} (\popt \log\log \popt))
        =
        O( n \log^{d-1} \popt )
    \end{equation*}
    time. Clearly, $\cup_i \Q_i$ is the desired piercing set.
\end{proof}

We can use the algorithm of \thmref{pierce:alg}, for the piercing
algorithm. For this choice,
\begin{equation*}
    \TPierce(m) = O( m^{d/2} \log^{2d+3} m).
\end{equation*}

We then get an approximation algorithm with running time
\begin{equation*}
    O\pth{ \nRounds {\pth{\popt^{1-1/\nRounds} n^{1/\nRounds}}}^{d/2}\log^{O(d)} (n) +
       n \log^{d-1} \popt }.
\end{equation*}

We thus get our second main result.
\begin{theorem}
    \thmlab{main:2}%
    Let $\Boxes$ be a set of $n$ axis-aligned boxes in $\Re^d$, for
    $d\geq 2$, and let $\nRounds > 0$ be an integer. A piercing set of
    $\Boxes$ of size $O(\nRounds d^2 \popt \log\log\popt )$ can be
    computed in
    \begin{equation*} %
        O\pth{ \nRounds \popt^{d/2-d/2\nRounds}  n^{d/2\nRounds}  \log^{O(d)} (n) +
           n \log^{d-1} \popt }.
    \end{equation*}
    expected time, where $\popt = \PNX{\Boxes}$ is the size of the
    optimal piercing set.
\end{theorem}

The above algorithm provides a trade-off between the approximation
factor and the running time. It readily leads to a near linear time
algorithm if the piercing set is sufficiently small. For example, by
choosing $\nRounds=d$, we obtain the following:

\begin{corollary}
    \corlab{near:linear:time}%
    Let $\Boxes$ be a set of $n$ axis-aligned boxes in $\Re^d$ for
    some fixed $d\geq 2$, and assume $\PNX{\Boxes} = O(
    n^{1/(d-1)})$. Then, a piercing set of $\Boxes$ of size
    $O(d^3 \popt \log\log\popt )$ can be computed in
    $O(n \log^{O(d)} n)$ expected time.
\end{corollary}

If the piercing set is slightly sublinear, the above leads to an
approximation algorithm with running time $O(n \log n)$.

\begin{corollary}
    \corlab{2d:linear:time}%
    Let $\Boxes$ be a set of $n$ axis-aligned rectangles in $\Re^2$
    for some fixed $d\geq 2$, and assume that it can be pierced by
    $\popt = O( n/ \log^{15} n)$ points. Then, a piercing set of
    $\Boxes$ of size $O( \popt \log\log\popt )$ can be computed in
    $O(n \log \popt)$ expected time.
\end{corollary}

\begin{proof}
    Pick a random sample $\Sample \subseteq \Boxes$ of size
    $O(n /\log^7n)$. The algorithm of \thmref{pierce:alg} yields in
    $O(n)$ time a piercing set $\Q$ for $\Sample$, of size
    $u = O( \popt \log \log \popt)$.  Preprocess $\Q$ for orthogonal
    range emptiness queries -- this takes $O( \popt \log^2 \popt)$ time,
    and one can decide if a rectangle is not pierced by $\Q$ in
    $O( \log \popt )$ time.  \lemref{sampling} implies that at most
    $\delta n$ rectangles unpierced by $\Q$, where
    \begin{equation*}
        \delta = \frac{ \popt  \log^8 n }{n}.
    \end{equation*}
    Namely, the \si{unhit} set has size $\delta n = O( n / \log^7
    n)$. Running time algorithm of \thmref{pierce:alg} on this set of
    rectangles, takes $O(n)$ time, and yields a second piercing set
    $\Q'$ of size $O( \popt \log \log \popt)$. Combining the two sets
    results in the desired piercing set.
\end{proof}

Alternatively, we can use the algorithm of \thmref{thm:fast-pierce}, for the piercing
algorithm. For this choice,
\begin{equation*}
    \TPierce(m) = O( m^{\log d+1}).
\end{equation*}

We then get an approximation algorithm with running time
\begin{equation*}
    O\pth{ \nRounds {\pth{\popt^{1-1/\nRounds} n^{1/\nRounds}}}^{\log
          d+1 }
       \log^{O(d)} n +
       n \log^{d-1} \popt }.
\end{equation*}

\begin{theorem}
    \thmlab{faster:bad:approx}%

    Let $\Boxes$ be a set of $n$ axis-aligned boxes in $\Re^d$, for
    $d\geq 2$, and let $\nRounds > 0$ be an integer. A piercing set of
    $\Boxes$ of size $O(\nRounds 2^{O(d)} \popt \log\log\popt )$ can be
    computed in
    \begin{equation*} %
        O\pth{ \nRounds \popt^{\log d+1 -\frac{\log
                 d+1}{\nRounds}}
           n^{\frac{\log d+1}{\nRounds}}  \log^{O(d)} n +
           n \log^{d-1} \popt }.
    \end{equation*}
    expected time, where $\popt = \PNX{\Boxes}$ is the size of the
    optimal piercing set.
\end{theorem}

By choosing $\nRounds=2\log d+2$, we obtain the following:

\begin{corollary}
    \corlab{near:linear:b:aprx}%
    Let $\Boxes$ be a set of $n$ axis-aligned boxes in $\Re^d$ for
    some fixed $d\geq 2$, assuming
    $\PNX{\Boxes} = O( n^{1/(2\log d+1)})$. Then, a piercing set of
    $\Boxes$ of size $O( 2^{O(d)} \popt \log\log\popt )$ can be
    computed in $O(n \log^{O(d)} n )$ expected time.
\end{corollary}

\section{Dynamic algorithm for piercing}
\seclab{dynamic}

We present a data structure for maintaining a near-optimal piercing
set for a set $\Boxes$ of boxes in $\Re^2$ as boxes are inserted into
or deleted from $\Boxes$.  By adapting the multi-round sampling based
algorithm described in \secref{iteralg}, we obtain a Monte Carlo
algorithm that maintains a piercing set of $O(\popt\log\log\popt)$ size
with high probability and that can update the piercing set in
$O^*(n^{1/2})$ amortized expected time per update.  (The $O^*()$
notation hides polylogarithmic factors). The update time can be
improved to $O^*(n^{1/3})$ if $\Boxes$ is a set of squares in $\Re^2$.

\paragraph{Overview of the Algorithm}

We observe that the size of the optimal piercing set changes by at
most one when a box is inserted or deleted. We periodically
reconstruct the piercing set using a faster implementation of the
multi-round sampling based algorithm in \secref{iteralg}, as described
below.  More precisely, if $s$ is the size of the piercing set
computed during the previous reconstruction, then we reconstruct the
piercing set after $\ceil{s/2}$ updates. To expedite the
reconstruction, we maintain $\Boxes$ in a data structure as
follows. We map a box $\bb=[a_1,a_2]\times [b_1,b_2]$ to the point
$\bb^*=(a_1,b_1,a_2,b_2)$ in $\Re^4$, and let $\Boxes^*$ be the
resulting set of points in $\Re^4$. We store $\Boxes^*$ into a
4-dimensional dynamic range tree $T$, which is a 4-level tree. Each
node $v$ of $T$ is associated with a \emph{canonical subset}
$\Boxes_{v}^*\subseteq\Boxes^*$ of points. Let $\Boxes_v$ be the set
of boxes corresponding to $\Boxes_v^*$. For a box $\square$ in
$\Re^4$, $\square\cap \Boxes^*$ can be represented as the union of
$O(\log^4 n)$ canonical subsets, and they can be computed in
$O(\log^4 n)$ time. The size of $T$ is $O(n\log^4 n)$, and it can be
updated in $O(\log^4 n)$ amortized time per insertion/deletion of
point. See \cite{bcko-cgaa-08}.

Between two consecutive reconstructions, we use a lazy approach to
update the piercing set, as follows: Let $\P$ be the current piercing
set. When a new box $\bb$ is inserted, we insert it into $T$.  If
$\P\cap \bb=\emptyset$, we choose an arbitrary point $p$ inside $\bb$
and add $p$ to $\P$.  When we delete a box $\bb$, we simply delete
$\bb^*$ from $T$ but do not update $\P$. If $\ceil{s/2}$ updates have
been performed since the last reconstruction, we discard the current
$\P$ and compute a new piercing set as described below.

We show below that a piercing set of size $s=O(\popt\log\log \popt)$ of
$\Boxes$ can be constructed in
\begin{equation*}
    O^*\pth{(\popt n)^{1/2}+\min\{\popt^2,n\}}
\end{equation*}
expected time, where $\popt$ is the size of the optimal piercing set
of $\Boxes$. This implies the amortized expected update time is
$O^*\pth{(n/\popt)^{1/2}+\min \{\popt, n/\popt\}}$, including the time
spent in updating $T$.  The second term is bounded by $n^{1/2}$, so
the amortized expected update time is $O^*(n^{1/2})$.

\paragraph{Reconstruction algorithm} Here is how we construct the
piercing set of boxes in $\Re^2$. Let $\Boxes$ be the current set of
boxes. We follow the algorithm in \thmref{main:2} and set the number
of rounds to 2. More precisely, perform an exponential search on the
value of $k$, the guess for the size of the optimal piercing set,
every time we reconstruct the piercing set. For a fixed $k$, the
reconstruction algorithm consists of the following steps:

\begin{compactenumI}
    \item Choose a random sample $\Boxes_1$ of $\Boxes$ of size
    $r=c_1(kn)^{1/2}$, where $c_1$ is a suitable constant.
    \item Construct a piercing set $\P_1$ of $\Boxes_1$ of size
    $s=O(k\log\log k)$ in $O^*(r)$ time using the algorithm in
    \secref{improved}.

    \item Compute $\Boxes_2\subseteq\Boxes$, the subset of boxes that
    are not pierced by $\P_1$. If $|\Boxes_2|>c_2{(kn)}^{1/2}$, where
    $c_2$ is a suitable constant, we return to Step 1. As described
    below, this step can be computed in
    $O^*(\min\{k^2,n\}+{(kn)}^{1/2})$ time.

    \item Compute a piercing set $\P_2$ of $\Boxes_2$, again using the
    algorithm in \secref{improved}.

    \item Return $\P_1\cup\P_2$.

\end{compactenumI}

The expected running time of this algorithm is
$O^*(\min\{k^2,n\}+{(kn)^{1/2}})$, as desired.

\paragraph{Computing $\Boxes_2$}
We now describe how to compute $\Boxes_2$ efficiently using $T$. If
$\popt \geq n^{1/2}$, then we simply preprocess $\P$ into a
2-dimensional range tree in $O(s\log s)$ time. By querying with each
box in $\Boxes$, we can compute $\Boxes_2$ in $O(n\log n)$ time
\cite{bcko-cgaa-08}. The total time spent is $O(n\log n)$. So assume
$\popt<n^{1/2}$. For a point $p = (x_p, y_p) \in \Re^2$, let
$Q_p \subset \Re^4$ be the orthant
\begin{equation*}
    Q_p = \Set{(x_1, x_2, x_3, x_4)}{ x_1 \leq x_p, x_2 \leq y_p,
       x_3 \geq x_p, x_4 \geq y_p \bigr.}.
\end{equation*}
 Then, a box $\bb \subset \Re^2$
contains $p$ if and only if $\Boxes^* \in Q_p$. Therefore, an input
box $\bb\in\Boxes$ is not pierced by $\P_1$ if
$\bb^* \not\in \bigcup\nolimits_{p\in\P} Q_p$. Let
$\K = \Re^4\setminus \bigcup\nolimits_{p\in\P} Q_p$. It is well known
that the complexity of $\K$ is $O(s^2)$. Furthermore, $\K$ can be
partitioned into $O(s^2)$ boxes with pairwise-disjoint interiors, as
follows.

Let $\Q=\Set{Q_p}{ p \in \P}$, and let
$\hat{\P} = \Set{(x_p, y_p, x_p, y_p)}{ p \in \P}$ be their corners.
We sort $\hat{\P}$ by the $x_4$-coordinates of its points. Let
$\Delta$ be an $x_4$-interval between the $x_4$-coordinates of two
consecutive points of $\hat{\P}$. For any value $a \in \Delta$, the
cross-section $\Q_a$ of $\Q$ with the hyperplane $h_a: x_4 = a$ is a
collection of $s$ 3-dimensional octants, and $\K_a = \K \cap h_a$ is
the complement of the union of $\Q_a$.  Furthermore, the cross-section
$\K_a$ remains the same for any value of $a\in\Delta$. It is well
known that the complexity of $\K_a$ is $O(s)$, and that it can be
partitioned into a set $\R_a$ of 3-dimensional boxes in $O^*(s)$ time.
Hence, we can partition $\K$ inside the slab $\Re^3\times\Delta$ by
the set $\R_{\Delta} = \Set{R \times \Delta}{ R \in \R_{a}}$. By
repeating this procedure for all $x_4$-intervals between two
consecutive points of $\hat{\P}$, we partition $\K$ into a family $\R$
of $O(s^2)$ boxes.

Next, we query $T$ with each box $R \in \R$. The query procedure
returns a set $V_{R}$ of $O(\log^4 n)$ nodes of $T$ such that
$\Boxes^*\cap R = \bigcup\nolimits_{v\in V_{R}} \Boxes_v^*$. We thus
obtain a set $V$ of $O(s^2 \log^4 n)$ nodes of $T$ such that
$\Boxes_2^* = \bigcup\nolimits_{v\in V} \Boxes_v^*$. If
$\sum\nolimits_{v\in V}|\Boxes_v^*| \leq c_2{(kn)}^{1/2}$, we return
$\bigcup\nolimits_{v\in V}\Boxes_v$ as $\Boxes_2$. Otherwise, we return
NULL. The total time spent by this procedure is
$O^*(\min\{n,k^2\}+(kn)^{1/2})$.  Putting everything together we
obtain the following.

\begin{theorem}
    \thmlab{dynamic:rect}
    A set $\Boxes$ of $n$ boxes in $\Re^2$ can be stored in a data
    structure so that a piercing set of $\Boxes$ of size
    $O(\popt\log\log\popt)$ can be maintained with high probability
    under insertion and deletion of boxes with amortized expected time
    $O(n^{1/2}\polylog n )$ per insertion or deletion, where $\popt$
    is the piercing number of $\Boxes$.
\end{theorem}

\paragraph{Dynamic algorithm for squares in 2D}

If we have squares instead of boxes, then the reconstruction time
reduces to $O^{*}(\popt^{2/3}n^{1/3})$, which leads to an amortized
update time of $O^{*}((n/\popt)^{1/3}) = O^{*}( n^{1/3})$. We proceed
in a similar manner as before. There are two differences.  First, we
now choose a random sample of size $O(k^{2/3}n^{1/3})$, and the
algorithm works in three rounds.  After the first round, we have a
piercing set $\P_1$ of size $O(\popt\log\log\popt)$, and we need to
represent the set of squares not pierced by $\P_1$ as $O^*(\popt)$
canonical subsets, so that we can choose a random sample $\Boxes_2$
from this subset of squares. After the second round, we have a
piercing set $\P_2$ of $\Boxes_2$ of size
$O(\popt\log\log\popt)$. Finally, we find the subset
$\Boxes_3\subseteq \Boxes$ of squares not pierced by $\P_1\cup\P_2$
and compute a piercing set of $\Boxes$. It suffices to describe how we
compactly represent the set $\Boxes_2$.

We map a square, which is centered at a point $c$ and of radius (half
side length) $a$, to the point $\bb^*=(c,a)\in\Re^3$. A point
$p = (x_p, y_p) \in \Re^2$ is now mapped to the cone
\begin{equation*}
    C_p=\Set{(x,y,z)\in\Re^3}{\dY{(x,y)}{p}_{\infty}\leq
       z}
\end{equation*}
with the square cross-section and $p$ its apex. The set $C_p$ is the
graph of the $L_{\infty}$-distance function from $p$. It is easily
seen that $p \in \bb$ if and only if $\bb^* \in C_p$. Hence, a box
$\bb\in\Boxes_2$ if it lies below all the cones of
$\C=\Set{C_p}{p\in\P_1}$. Using a 3D orthogonal range-searching data
structure, we can compute $\Boxes_2$ as the union of $O^*(k)$
canonical subsets. We omit the details from here and obtain the
following.

\begin{theorem}
    \thmlab{dynamic:square}
    A set $\Boxes$ of $n$ squares in $\Re^2$ can be stored in the data
    structure described above so that a piercing set of $\Boxes$ of
    size $O(\popt\log\log\popt)$ can be maintained with high
    probability under insertion and deletion of boxes in
    $O(n^{1/3}\polylog n )$ amortized expected time per insertion or
    deletion.
\end{theorem}

\section{Data structure for maintaining weights of boxes}
\seclab{DS}

In this section, we describe how to implement the data structure
needed for our algorithms whose specifications are given in
\defref{arr:ds1}. Let $\Boxes$ be a set of axis-aligned boxes in
$\Re^d$. For simplicity, we assume the boxes in $\Boxes$ are in
general position. Let $\AB$ be a set of \emphi{active} boxes
(initially empty) and let $\VX{\AB}$ be the set of vertices of the
arrangement $\ArrX{\AB}$. Let $\S\subseteq \Boxes$ be a multi-set of
\emphi{update} boxes (initially empty). For a point $p\in \VV$, recall
that the \emphi{doubling weight} of $p$ is defined to be
$\hwY{p}{\S} = 2^{\nrstY{\S}{p}}$, where let
$\rstY{\S}{p} = \Set{ \bb \in \S}{ p \in \bb}$ is the multi-set of all
boxes in $\S$ containing $p$. We require a data-structure that
supports the following operations:
\begin{compactenumI}
    \smallskip%
    \item $\weight (\bb)$: given a box $\bb$, computes $\hwY{\VX{\AB}\cap \bb}{\S}$.

    \smallskip%
    \item $\double(\bb)$: given a box $\bb \in \Boxes$, adds a copy of
    $\bb$ to the multi-set of update boxes $\S$.

    \smallskip%
    \item $\halve(\bb)$: given a box $\bb \in \Boxes$, removes a copy
    of $\bb$ from the multi-set of update boxes $\S$.

    \smallskip%
    \item $\sample$: returns a random point $p \in \VV(\AB)$ with
    probability ${\hwY{p}{\S}}/{\hwY{\VX{\AB}}{\S}}$.

    \smallskip%
    \item $\Insert (\bb)$: inserts $\bb$ into the set of active boxes
    $\AB$.

    \smallskip%
    \item $\Delete (\bb)$: removes $\bb$ from the set of active boxes
    $\AB$.

\end{compactenumI}

\smallskip%
\noindent%

For a set $Z \subseteq \Re^d$ and $i \in \IRX{d} = \{1,\ldots, d\}$, let
$\projY{Z}{x_i} = \Set{ x_i }{ (x_1, \ldots, x_d) \in Z} $ be the
\emphi{projection} of $Z$ to the $x_i$-axis. For a multiset $\S$ of
boxes in $\Re^d$ and $i \in \IRX{d}$, let
$\projY{\S}{x_i} = \Set{ \projY{\bb}{x_i} }{ \bb \in \S}$ be the
multiset of intervals resulting from projecting $\S$ on the
$x_i$-axis.

We first describe the data structure for the line in \secref{ds:line},
then for the plane in \secref{ds:plane}, and finally extend it to
higher dimensions in \secref{ds:higher}.

\subsection{Data structure in 1D}
\seclab{ds:line}

A segment tree with minor tweaking provides the desired data-structure
for 1D, as described next.

The set of vertices $\VV$ is simply the endpoints of the
intervals of $\Boxes$. We construct the segment tree $\Tree$ of
$\Boxes$, where the endpoints of $\Boxes$ are stored at the leaves of
$\Tree$ (thus, a node $v$ in $\Tree$ corresponds to an ``interval''
which is the set of all endpoints stored in its subtree).  The idea is
to update the weights in a lazy fashion --- for completeness we
describe this in detail, but this is by now folklore.

We modify the segment tree so that each internal node $v$ has a total
weight $\wX{v}$ of all the vertices in its subtree, and a count
$\alpha(v)$.  In any given point in time, if $v$ has two children $x$
and $y$, we have that $\wX{v} = 2^{\alpha(v)} (\wX{x} + \wX{y})$. The
value of $\alpha(v)$ is pushed down to its children in a lazy
fashion. Specifically, whenever the algorithm traverses from a node
$v$ to one of its children, it adds the value of $\alpha(v)$ to
$\alpha(x)$ and $\alpha(y)$, and sets $\alpha(v)$ to zero. It
immediately also recomputes $\wX{x}$ and $\wX{y}$. Similarly, whenever
the traversal returns from a child, the weight of the parent node is
recomputed.

For the operation $\weight(\bb)$, the data-structure locates the
$O( \log n)$ nodes of $\Tree$ such that their (disjoint) union covers
the interval $\bb$, and it returns the sum of their weights (note,
that this propagates down the values of $\alpha(\cdot)$ at all the
ancestors of these nodes to them -- thus, all the ancestors have
$\alpha(\cdot)$ with value zero).  The operation $\double(\bb)$
(resp. $\halve(\bb)$) works in a similar fashion, except that the
data-structure increases (resp. decreases) the $\alpha(\cdot)$ of the
all the $O( \log n)$ nodes of $\Tree$ covering $\bb$ by one.

The operation \sample performs a random traversal down the tree. If
the traversal is at node $v$, with children $x$ and $y$, it continues
to $x$ with probability $\wX{x}/\wX{v}$, and otherwise it continues to
$y$. The function returns the endpoint stored in the leaf where the
traversal ends.

Finally, if we want to support \Insert and \Delete, we initially set
the weight of endpoints associated with each interval to zero. When an
interval $(a,b)\in \Boxes$ is inserted to the set of active intervals
$\AB$, we locate the leaves corresponding to $a$ and $b$ and set their
weights to one. We then traverse up the tree from the leaves
recomputing $\omega(\cdot)$ for the ancestors of either leaf. The
deletion sets the weight of the endpoints to zero and proceeds
analogously. Therefore, we obtain the following:

\begin{lemma}
    \lemlab{d:s:1:dim}%
    Let $\Boxes$ be a set of $n$ intervals in $\Re^1$. A
    data-structure of size $O(n\log n)$ can be constructed in
    $O(n\log n)$ time that supports each of the operations described
    in \defref{arr:ds} in $O(\log n)$ time.
\end{lemma}

\subsection{Data structure in the plane}
\seclab{ds:plane}
\paragraph{Basic idea}

Let $\Boxes$ be the set of input rectangles, and let $\cell$ be an
axis-parallel box that does not contain a vertex of a rectangle of
$\Boxes$ in its interior. Any edge $e$ of a rectangle $\bb \in \Boxes$
that intersects the interior of $\cell$, cuts it into two pieces by
the line supporting $e$. Let $X$ (resp. $Y$) be the set of all the
$x$-values (resp. $y$-values) of all edges that are orthogonal to the
$x$-axis (resp. $y$-axis) and that intersect $\cell$. Clearly, the
vertices of $\VV$ in the interior of $\cell$ are formed by the
Cartesian product $X \times Y $ (we refer to this set somewhat
informally, as a \emphw{grid}). Weights on such a grid of vertices can
be maintained implicitly by maintaining weights on the sets $X$ and
$Y$ separately (using the one dimensional data-structure). It is easy
to verify that all the operations of \defref{arr:ds} decompose into
these one-dimensional data-structures.

Thus, the data-structure is going to decompose the arrangement of
$\Boxes$ into cells, where the vertices inside each cell would be
represented by such a grid. Importantly, each box would appear
directly in only $O(\sqrt{n})$ grids. A technicality is that update
box might contain many such cells in their interior (and thus their
grids are contained inside the update box). Fortunately, by storing
these cells in an appropriately constructed tree, such updates could be handled by implicit updates on the nodes
of this tree.

\paragraph{Partitioning scheme}
Let $\Boxes$ be a set of axis-aligned boxes in $\Re^2$, and let
$\square$ be a rectangle that contains all boxes of $\Boxes$ in its
interior. We first partition $\square$ into $\ceil{ 2\sqrt{n} }$
vertical slabs, i.e., by drawing $y$-parallel edges, so that each slab
contains at most $\sqrt{n}$ vertical edges of boxes in $\Boxes$. Next,
we further partition each slab $\sigma$ into $O(\sqrt{n})$ rectangles
by drawing horizontal edges as follows. If a corner (vertex) $\xi$ of
a box of $\Boxes$ lies in $\sigma$, we partition $\sigma$ by drawing a
horizontal edge passing through $\xi$. Finally, if a rectangle
$\varrho$ in the subdivision of $\sigma$ intersects more than
$\sqrt{n}$ horizontal edges of $\Boxes$, we further partition
$\varrho$ into smaller rectangles by drawing horizontal edges, so that
it intersects at most $\sqrt{n}$ horizontal edges. Let $\Partition$ be
the resulting rectangular subdivision of $\square$. Observe that
$\Partition$ has $O(n)$ cells. By construction, no vertex of $\Boxes$
lies in the interior of a cell of $\Partition$, the boundary of any
box intersects $O(\sqrt{n})$ cells of $\Partition$, and each cell
$\cell$ intersects at most $2\sqrt{n}$ edges of $\Boxes$.  For a cell
$\cell$ of $\Partition$, let $\VV \cap {\cell}$ be the set of all
vertices of $\VV = \VX{\Boxes}$ that lie in the \emph{interior} of
$\cell$. Observe that we have
$\VV \cap {\cell} = \projY{(\VV \cap {\cell})}{x} \times \projY{(\VV
   \cap {\cell})}{y}$.

\paragraph{Weights decompose in a cell}
Consider a multiset $\S_{\cell} \subseteq \Boxes$ that intersect the
interior of a cell $\cell$. The multiset $\S_{\cell}$ can be
partitioned into two multisets of intervals:
\begin{equation*}
    \XX=
    \Set{ \projY{\bb}{x}}{ \bb \in \S_{\cell} \text{ and  }
       \projY{\cell}{x} \not\subset
       \projY{\bb}{x}}
    \qquad\text{and}\qquad%
    \YY=
    \Set{ \projY{\bb}{y}}{ \bb \in \S_{\cell} \text{ and  }
       \projY{\cell}{y} \not\subset
       \projY{\bb}{y}}.
\end{equation*}
Note, that a box $\rect \in \S_{\cell}$ contributes an interval (more
formally, its multiplicity of intervals) to exactly one of these sets,
as no box of $\S_{\cell}$ can have a vertex in the interior of $\cell$.

\begin{lemma}
    \lemlab{2dim}%
    Let $\cell$ be a cell of $\Partition$, $\S_{\cell} \subseteq \Boxes$ a
    multiset of rectangles whose boundaries intersect the interior of
    $\cell$, and let $\rect$ be a rectangle.  For
    $U = \VV \cap \cell \cap \rect$,
    $X = \projY{U}{x}$ and
    $Y = \projY{U}{y}$, we have that
    \begin{math}
        \hwY{U}{\S_{\cell}} =%
        \hwY{\bigl.X}{\XX} \cdot \hwY{\bigl.Y}{\YY},
    \end{math}
\end{lemma}

\begin{proof}
    Since no box $\rect \in \S_{\cell}$ has a vertex in the interior of
    $\cell$, it follows that either the vertical or the horizontal
    edges of $\rect$ intersect $\cell$ (but not both). As such, we
    have $U = X \times Y$.  Furthermore, for $p =(x,y) \in U$, we have
    that $\nrstY{p}{\S_{\cell}} = \nrstY{x}{\XX} + \nrstY{y}{\YY}$, which
    implies that $\hwY{p}{\S_{\cell}} = \hwY{x}{\XX}\cdot \hwY{y}{\YY}$.  As
    such, we have
    \begin{align*}
        \hwY{U}{\S_{\cell}}
      &=%
        \sum_{p \in U}  \hwY{p}{\S_{\cell}}
        =%
        \sum_{(x,y) \in X \times Y}  \hwY{x}{\XX}\cdot \hwY{y}{\YY}
      =%
      \sum_{x \in X} \hwY{x}{\XX} \cdot\sum_{y \in Y}\hwY{y}{\YY}
      \\&
      =
      \hwY{X}{\XX}
      \cdot
      \hwY{Y}{\YY}.
    \end{align*}
\end{proof}

\begin{figure}[h]
    \centering \includegraphics{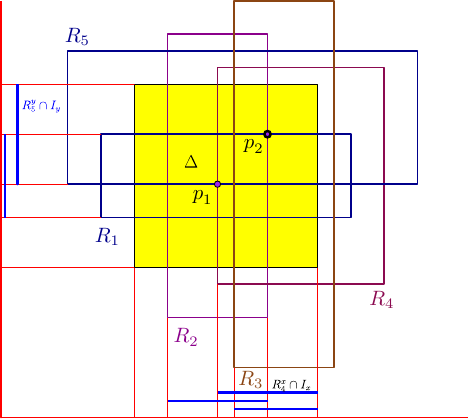}

    \caption{$\Delta=I_x \times I_y$ denotes a cell in
       $\Partition$. $\S_{\cell}=\{R_1,R_2,R_3,R_4\}$ intersects $\Delta$. Let
       $R_i=R_i^x\times R_i^y$ for $i\in\{1,2,3,4\}$.
       $\MX_{\Delta}=\{R_2^x, R_3^x, R_4^x\}$,
       $\MY_{\Delta}=\{R_1^y, R_5^y\}$. Observe that
       $\hwY{\{p_1\}}{\S_{\cell}} =
       \hwY{\projY{\{p_1\}}{x}}{\MX_{\Delta}} \cdot
       \hwY{\projY{\{p_1\}}{y}}{\MY_{\Delta}}=16$ and
       $\hwY{\{p_2\}}{\S_{\cell}} =
       \hwY{\projY{\{p_2\}}{x}}{\MX_{\Delta}} \cdot
       \hwY{\projY{\{p_2\}}{y}}{\MY_{\Delta}}=32$. }
    \figlab{cell}
 \end{figure}

 \paragraph{Tree data structure}
 We construct a balanced binary tree $\Tree$ storing the cells of
 $\Partition$ in the leaves -- the top levels form a balanced binary
 tree over the slabs, say from left to right. A ``leaf'' that stores a
 slab is then the root of a balanced binary tree on the cells within
 the slab, with cells ordered from bottom to top.  Each node $u$ of
 $\Tree$ is associated with a rectangle $\Box_u$. The root is
 associated with $\square$ itself and each leaf is associated with the
 corresponding cell of $\Partition$. For an internal node $u$ with
 children $w$ and $z$, $\Box_u=\Box_w\cup\Box_z$. Let
 $\VV_u=\VV\cap \intX{\Box_u}$. If $u$ is an internal node with $w$
 and $z$ as children and $e$ being the common edge of $\Box_w$ and
 $\Box_z$ (i.e.  $e=\Box_w\cap \Box_z$), then
 $\VV_u=\VV_w\cup \VV_z\cup (\VV\cap e)$. If $e$ does not lie in the
 horizontal edge of a rectangle of $\Boxes$ then
 $\VV\cap e=\emptyset$. If $e\in \partial \rect$ for some rectangle
 $\rect \in\Boxes$ then $\VV_u\cap e$ is the set of intersection
 points of $e\cap \partial \rect$ with the $x$-edges of $\Boxes$ that
 intersect $e$.

 \begin{figure}[htp]
     \phantom{}%
     \hfill%
     \includegraphics[page=1]{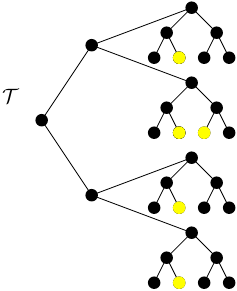}%
     \hfill%
     \includegraphics[page=2]{figs/tree}%
     \hfill%
     \phantom{}%

    \caption{The figure represents the partition $\Partition$ and the tree
       $\Tree$ built on $\Partition$. $\sigma$ is a
       strip of $\Partition$, $\Delta$ is
       a cell of $\sigma$. $\Tree_{\sigma}$ is the tree built on the
       cells of $\sigma$. $R$ is a rectangle in $\Boxes$. The highlighted
       nodes represent $\C(\bb)$. ( Note that the axes are flipped.) }
    \figlab{part}
\end{figure}

Let $\rect = \rect_x \times \rect_y$ be an arbitrary rectangle. The
rectangle $\rect$ is \emphi{long}, at a node $u$, if
$\Box_u\subseteq \rect $, and \emphi{short} if
$\Box_u\cap \partial \rect \neq \emptyset$. A node $u$ is
\emphi{canonical} for $\rect$, if ({\romannumeral 1}) $u$ is a leaf
and $\rect$ is short at $u$, or ({\romannumeral 2}) if $\rect$ is long
at $u$ but short at $p(u)$, the parent of $u$ in $\Tree$. Let
$\C(\rect)$ denote the set of canonical nodes for $\rect$ and let
$\C^*(\rect)$ be the set of ancestors of nodes in $\C(\rect)$
including $\C(\rect)$ itself. By the construction of $\Partition$,
$\rect$ is short at $O(\sqrt{n})$ leaves of $\Tree$, therefore
$|\C(\rect)|=|\C^*(\rect)|=O(\sqrt{n}\log n)$.

For simplicity, in the subsequent discussion we assume that all boxes
are active. Notice that the construction of $\Tree$ is not affected by
the set of active boxes. Later, we describe how to insert or delete
from the set of active boxes $\AB$ while preserving the invariants of
the data structure.  Recall that $\S$ is the multi-set of update
boxes. For a node $u\in \Tree$, let $\S_u\subseteq \S$
(resp. $\L_u\subseteq \S$) be the set of update boxes that are short
(resp. long) at $u$. Note that $u$ is a canonical node for all boxes
in $\L_u\setminus \L_{p(u)}$ and also for $\S_u$ if $u$ is a leaf. We
maintain two values at each node $u$ of $\Tree$:
\begin{compactenumi}
    \smallskip%
    \item $\lambda_u = \cardin{\L_u\setminus \L_{p(u)}}$: the number
    of boxes of $\S$ that are long at $u$ but short at $p(u)$;
    and

    \smallskip%
    \item $\omega_u=\hwY{\VV_u}{\S_u}$: the total weight of points in
    $\VV_u$ with respect to the update boxes that are short at
    $u$.
\end{compactenumi}
\smallskip%
The following inequalities follow easily from the definitions for any
rectangle $\rect$:
\begin{align}
  \eqlab{1}
  &\hwY{\VV_u}{\S}= \hwY{\VV_u}{\S_u}
    \times\hwY{\VV_u}{\L_u}
    =
    \omega_u\times \prod\nolimits_{v\in \anc(u)} 2^{\lambda_v}
  \\&
  \eqlab{2}
  \hwY{\VV_u\cap\rect}{\S}
  =%
  \hwY{\VV_u\cap \rect}{\S_u} \times \prod\nolimits_{v \in \anc(u)} 2^{\lambda_v},
\end{align}
where $\anc(u)$ is the set of ancestors of $u$, including $u$
itself. Let $u$ be an internal node of $\Tree$ with children $y$ and
$z$, and let $e_u=\Box_y\cap\Box_z $ be the common edge shared by
$\Box_y$ and $\Box_z$. Then, for any rectangle $\rect$, we have
\begin{align}
  \eqlab{3}
  &\omega_u=\omega_y2^{\lambda_y}+\omega_z 2^{\lambda_z}+\hwY{\VV\cap e_u}{\S_u}
  \\
  \eqlab{4}
  \text{and}\qquad%
  &\hwY{\VV_u\cap \rect}{\S_u}
    =
    \hwY{\VV_w\cap \rect}{\S_y}2^{\lambda_y}
    + \hwY{\VV_z\cap \rect}{\S_z} 2^{\lambda_z}+\hwY{\VV\cap e_u}{\S_u},
\end{align}
If $e_u$ is a vertical edge then $\VV\cap e_u=\emptyset$, so the last
term is positive only if $e_u$ is a horizontal edge of a cell in
$\Partition$.

By \Eqref{3}, $\omega_u$ can be maintained recursively if the value at
the leaves of $\Tree$ is known, and one can compute
$\hwY{\VV\cap e_u}{\S_u}$ efficiently. To this end, we use the ``lazy
segment tree'' data structure from \secref{ds:line}, as secondary data
structures, to maintain $\omega_u$ at the leaves. The same data
structure is also used to maintain $\hwY{\VV \cap e_u}{\S_u}$.

Let $z$ be a leaf of $\Tree$, and let $\Box_z$ be the cell in
$\Partition$ corresponding to $z$.  Let
$X = \projY{\pth{\Box_z \cap \VV}}{x}$ and
$Y = \projY{\pth{\Box_z \cap \VV}}{y}$. Note that,
$|X|=|Y|=O(\sqrt{n})$ because only $O(\sqrt{n})$ rectangle boundaries
intersect the interior of $\Box_z$. Then we construct lazy segment
trees $\Psi_z^x$ and $\Psi_z^y$ on $X$, and $Y$ respectively. Recall
that $\S_z\subseteq\S$ is the set of update rectangles that are short
at $z$. We can decompose $\S_z$ into the following multi-sets of
intervals:

\begin{equation*}
     \XXz=
    \Set{ \projY{\bb}{x}}{ \bb \in \S_{z} \text{ and  }
       \projY{(\Box_z)}{x} \not\subset
       \projY{\bb}{x}}
    \qquad\text{and}\qquad%
    \YYz=
    \Set{ \projY{\bb}{y}}{ \bb \in \S_{z} \text{ and  }
       \projY{(\Box_z)}{y} \not\subset
       \projY{\bb}{y}}. %
\end{equation*}

By \lemref{2dim},
\begin{math}
    \omega_z%
    =%
    \hwY{\VV_z}{\S_z}%
    =%
    \hwY{X}{\XXz}\cdot \hwY{Y}{\YYz}.
\end{math}
Hence, we can maintain $\omega_z$ by maintaining $\Psi_z^x$ and
$\Psi_z^y$. Furthermore, for each internal node $u$ with children
$v,w$ such that $\Box_u$ lies in a slab of $\Partition$, we also have
a lazy segment tree $\Psi_u^e$ for handling the vertices in $\VV \cap e$
where $e=\Box_v\cap\Box_w$.  We omit the details from here.

This completes the description of the data structure. The
tree $\Tree$, as well as the secondary data structures at all leaves
and relevant internal nodes of $\Tree$ can be constructed and
initialized in $O(n^{3/2}\log n)$ time.

\paragraph{Operations on \TPDF{$\Crates$}{T}}

We now briefly describe how to perform the operations mentioned
in \secref{mwu:basic} using the data structure.

\begin{compactitem}
    \item $\weight(\rect)$: Given a rectangle
    $\rect=\rect_x\times \rect_y$, the goal is to return
    $\hwY{\VV_{\rho}\cap \rect}{\S}$ where $\rho$ is the root of
    $\Tree$. We start at the root node and visit $\Tree$ in a top-down
    manner to find the set of canonical nodes in $\Tree$ with respect
    to $\rect$.  Then, starting from the canonical set of nodes we
    visit $\Tree$ bottom-up computing
    $\hwY{\VV_{u}\cap \rect}{\S_{u}}$ for each node $u$ we visit. In
    particular, for each internal node $u$ we visit, if $\rect$ is
    long at $u$ then $\hwY{\VV_u\cap \rect}{\S_u}=\omega_u$.  If
    $\rect$ is short at $u$, and $w, z$ are its children then,
    \begin{equation*}
        \hwY{\VV_u\cap \rect}{\S_u}
        =
         2^{\lambda_{w}}\hwY{\VV_w\cap \rect}{\S_{w}}
            +2^{\lambda_{z}}
            \hwY{\VV_z\cap \rect}{\S_{z}} + \hwY{\VV \cap e}{\S_u}
    \end{equation*}
    where $e=\Box_{w}\cap \Box_{z}$. We recursively compute each of
    the terms.  If $u$ is a leaf node, we compute
    $\hwY{\VV_u\cap \rect}{\S_u}$ by calling $\weight(\rect_x)$ and
    $\weight(\rect_y)$ on $\Psi_u^x$ and $\Psi_u^y$, respectively and
    taking their product.  \lemref{2dim} guarantees correctness in
    this case. In the case of the internal nodes, correctness follows
    from \Eqref{1}, \Eqref{2}, \Eqref{3}, and \Eqref{4}.

    \smallskip%

    \item $\double(\rect)$: Given a rectangle
    $\rect=\rect_x\times \rect_y\in \Boxes$, we first add a copy of
    $\bb$ to the set of update boxes $\S$. We then need to modify the
    data structure to ensure all the invariants are preserved. We
    start at the root of $\Tree$ and we traverse $\Tree$ in a top-down
    manner. If $\rect$ is long at $u$, we increment $\lambda_u$ by
    $1$. If $\rect$ is short at $u$ and $u$ is an internal node with
    children $w$ and $z$, and $e=\Box_{w}\cap\Box_{z}$, we first
    update $\Psi_u^e$ and then recurse on $w$ and $z$. If we reach a
    leaf $u$, and vertical (resp. horizontal) edges of $\rect$
    intersect $\Box_u$, we call $\double(\rect_x)$
    (resp. $\double(\rect_y)$) on $\Psi_u^x$ (resp. $\Psi_u^y$). We
    also update the values of $\omega_u$ for every node $u$ we
    visit. It is easy to verify that the operation preserves the
    invariants captured by \Eqref{1}, \Eqref{2}, \Eqref{3}, and
    \Eqref{4}.

    \smallskip%
    \item $\halve(\rect)$: Given a rectangle
    $\rect=\rect_x\times \rect_y\in \Boxes$, we proceed analogously to
    the way we do for $\double$. We start at the root of $\Tree$ and
    traverse in a top-down manner. If $\rect$ is long at $u$, we
    decrement $\lambda_u$ by $1$. If $\rect$ is short at $u$ and $u$
    is an internal node with children $w$ and $z$ and
    $e=\Box_{w}\cap\Box_{z}$, we first update $\Psi_u^e$ and then
    recurse on $w$ and $z$. If we reach a leaf $u$, and vertical
    (resp. horizontal) edges of $\rect$ intersect $\Box_u$, we call
    $\halve(\rect_x)$ (resp. $\halve(\rect_y)$) on $\Psi_u^x$
    (resp. $\Psi_u^y$).

    \smallskip%
    \item $\sample$: We start at the root of $\Tree$ and we traverse
    $\Tree$ in a top-down manner. For a node $u$ with children $w, z$,
    we pick one of $w, z$ and $e=\Box_{w}\cap \Box_{z}$ with
    probabilities
    $\hwY{\VV\cap \Box_{w}}{\S}, \hwY{\VV \cap \Box_{z}}{\S}$, and
    $\hwY{\VV\cap e}{\S}$ respectively. These probabilities can easily
    be computed using the secondary data structures and the values
    maintained at each node, as done in the $\weight$ operation. If
    $e$ is picked, we then use the $\sample$ operation on the
    secondary data structure $\Psi_z^e$ to sample a point from
    $\VV\cap e$. If $w$ or $z$ is picked we recurse on it. If at any
    stage we reach a leaf $l$, we pick $a$ and $b$ by calling
    $\sample$ on $\Psi_l^x$ and $\Psi_l^y$, respectively, and return
    the point $(a,b)$.

    \smallskip%

    \item $\Insert(\rect)$: Given a rectangle
    $\rect=\rect_x\times \rect_y$, we first add a copy of $\bb$ to the
    set of active boxes $\AB$. Although a new active box does not
    affect weights of existing vertices, it can add new vertices
    thereby making the value of $\omega_u$ inconsistent for a node
    $u$. We proceed in the same way we do for $\double$, except we do
    not modify the $\lambda_u$ values and use the $\Insert$ operation
    on the secondary data structures. We also update the values of
    $\omega_u$ for every node $u$ we visit.  It is easy to verify that
    the operation preserves the invariants captured by \Eqref{1},
    \Eqref{2}, \Eqref{3}, and \Eqref{4}.

    \item $\Delete(\rect)$: Given a rectangle
    $\rect=\rect_x\times \rect_y$, we first remove a copy of $\bb$ to
    the set of active boxes $\AB$. The rest of the procedure is the
    same as $\Insert$, except we use the $\Delete$ operation on the
    secondary data structures in the leaves.

\end{compactitem}

\paragraph{Running time} %

As discussed above, a rectangle $\rect$ has at most
$O(\sqrt{n} \log n)$ canonical nodes. Hence, to compute
$\weight(\cdot)$, $\double(\cdot)$, $\halve(\cdot)$, $\Insert(\cdot)$,
or $\Delete(\cdot)$ operations we need to visit $O(\sqrt{n} \log n)$
nodes.  Also, $O(\sqrt{n})$ secondary data structure operations are
performed each of which takes $O(\log n)$ time. Hence, the running
time for these operations is bounded by $O(\sqrt{n}\log n)$. For
$\sample$, it is clear that we visit $O(\log n)$ nodes and perform a
total of $O(1)$ secondary data structure operations. Hence, $\sample$
takes $O(\log n)$ time.

\subsection{Extending to higher dimensions}
\seclab{ds:higher}

The basic idea of 2D can be extended to higher dimensions. Namely, let
$\cell$ be an axis-aligned box that does not contain any $(d-2)$-face
of a box $\bb\in\Boxes$. If a facet $f$ of a box $\bb\in \Boxes$
intersects the interior of $\cell$ then $f$ splits $\cell$ into two
boxes by the hyperplane supporting $f$. For $i=1,...,d$, let $X_i$ be
the $x_i$-values of the facets of the boxes in $\Boxes$ that are
orthogonal to the $x_i$-axis and that intersect $\cell$. Then
$\VV(\Boxes)\cap \cell=X_1\times...\times X_d$. We thus extend the 2D
data structure to higher dimensions.

\paragraph{Partitioning scheme}

To extend the above data structure to $\Re^d$, for $d>2$, we
generalize the partition scheme described above. The partition scheme
works recursively on the dimensions -- first applying a hyperplane cut
every time $\sqrt{n}$ faces orthogonal to the first axis are
encountered. Now, we project on the boxes that intersect a slab into
its boundary hyperplane. We construct a partition on this hyperplane,
applying the construction in $\Re^{d-1}$, and then lift the partition
back to the slab. We now briefly describe the specifics.

Let $\square$ be a box containing all of $\Boxes$. A facet of a box is
an \emphi{$i$-side} if it is normal to the $x_i$-axis.  We first
partition $\square$ into $\ceil{ 2\sqrt{n}}$ slabs, by drawing
hyperplanes parallel to the $1$-sides of the boxes so that each slab
contains at most $\sqrt{n}$ $1$-sides of boxes in $\Boxes$.  These are
\emphw{level-1} slabs. Next, we further partition each level-1 slab by
drawing hyperplanes along the $2$-sides. If a corner (vertex) $\xi$ of
a box of $\Boxes$ lies in $\sigma$, we partition $\sigma$ by drawing a
hyperplane parallel to the $2$-side passing through $\xi$. Finally, if
a cell $\cell$ in the subdivision of $\sigma$ intersects more than
$\sqrt{n}$ $2$-sides of $\Boxes$, we further partition $\cell$ by
drawing hyperplanes parallel to the $2$-sides, so that it intersects
at most $\sqrt{n}$ $2$-sides. The above steps partition each level-1
slab $\sigma$ into $O(\sqrt{n})$ level-2 slabs. Next, we repeat the
previous steps for each level-2 slab $\gamma$, i.e. we partition it
using hyperplanes along the $3$-sides such that no cell in the
resulting subdivision contains more than $O(\sqrt{n})$ 3-sides of
boxes in $\Boxes$. We continue these steps on the resulting level-3
slabs and so on. For more details see \cite{oy-nubkm-91}.

Let $\Partition$ be the resulting partition. By construction,
$\Partition$ is a subdivision of $\square$ into boxes. We have
$|\Partition|=O(n^{d/2})$, no vertex of $\Boxes$ lies in the interior
of a cell of $\Partition$, the boundary of any box intersects at most
$O(n^{(d-1)/ 2})$ cells of $\Partition$, and each cell intersects at
most $O(\sqrt{n})$ faces of $\Boxes$. Borrowing the definition from
\cite{oy-nubkm-91}, for a cell $\cell \in \Partition$, a box
$\bb \in \Boxes$ is an \emphi{$i$-pile} with respect to $\cell$ if
$\partial \bb \cap \cell \neq \emptyset$ and for all $j\neq i$, the
$j$\th interval of $\bb$ spans the $j$\th interval of $\cell$. The
partition ensures that if a box $\bb \in \Boxes$, has
$\partial \bb\cap \intX{\cell} \neq \emptyset$ then it is an $i$-pile
with respect to $\cell$ for some $i$.

For a cell $\cell \in \Partition$, let $\Edges_\cell^i$ denote the set
of $i$-sides of $\Boxes$ that intersect $\cell$, clipped within
$\cell$. Let $Z_\cell^i=\projY{(\VV\cap\cell)}{x_i}$ be the set of
$x_i$-coordinates of the $i$-sides in $\Edges_\cell^i$.  Let
$\VV_\cell\subseteq \VV$ be the subset of vertices that lie in the
interior of $\cell$. Then
\begin{math}
    \VV_\cell =%
    \Set{e_1 \cap \cdots \cap e_d}{ e_1\in \Edges_\cell^1,\ldots,
       e_d\in \Edges_\cell^d }.
\end{math} The following is the straightforward extension of \lemref{2dim} to
higher dimensions.

\begin{lemma}
    \lemlab{high}
    Let $\cell$ be a cell of $\Partition$, let $\S_{\cell}\subseteq \Boxes$ be
    a multiset of boxes whose boundaries intersect $\cell$, let $ \I^i_{\cell}=
    \Set{ \projY{\bb}{x_i}}{ \bb \in \S_{\cell} \text{ and  }
       \projY{\cell}{x_i} \not\subset
       \projY{\bb}{x_i}}$
    be the multiset of $x_i$-projections of boxes in $\S_{\cell}$ whose
    $i$-sides intersect $\cell$. For an arbitrary box $\bb$, we have
    \begin{equation}
        \hwY{\VV_\cell\cap \bb}{\S_{\cell}}
        =%
        \prod\nolimits_{i=1}^d \hwY{Z_\cell^i\cap \projY{\bb}{x_i}}{\I^i_{\cell}}.
    \end{equation}
\end{lemma}

We build a $d$-dimensional tree $\Tree$ on $\Partition$ analogously to
the 2D case. Let $z$ be a leaf of $\Tree$, and let
$\Box_z$ be the cell corresponding to $z$ . We construct a segment tree
$\Psi_z^i$ on $Z_{\Box_z}^i$ for each $i\in\{1,..,d\}$. The
operations of the data structure as well as the analysis are identical
to the 2D case and hence we omit the details. We get the following
lemma.

\begin{lemma}
    \lemlab{data-structure}%
    Let $\Boxes$ be a set of $n$ axis-aligned rectangles in $\Re^d$. A
    data-structure can be constructed in $O(n^{(d+1)/2}\log n)$ time
    that supports every operation specified in \defref{arr:ds}.
\end{lemma}

\section{Conclusions}
\seclab{conclusions}

In this paper, we presented static and dynamic algorithms for piercing
boxes by points. We proposed two types of static algorithms: One where
the approximation ratio has a polynomial dependence on $d$ but the
runtime is expensive, and the other where the runtime is faster but
there is an exponential dependence on $d$.  As mentioned in the
introduction, the recent result by Bhore and Chan \cite{bc-fsdaa-25} gave an
$O(n^{1+\delta})$ time algorithm that computes a piercing set of size
$d^{O(d)} \popt \log \log \popt$, where $\popt$ is the size of the
optimal piercing set.  This leaves open the question of whether both
conditions can be satisfied together, i.e., are there near-linear
algorithms with
\begin{math}
    O(\poly(d) \log \log \popt)
\end{math}
approximation ratio for piercing? Given the close connection between
piercing and finding the deepest point in an arrangement of
rectangles, this seems unlikely. An interesting future direction would
be to consider this connection further. Another line of work could
involve expanding our technique for reducing the candidate point set
to design faster algorithms for piercing other geometric objects.

\BibTexMode{%
   \bibliographystyle{alpha}%
   \bibliography{rect_net}%
}%

\BibLatexMode{\printbibliography}

\end{document}